\documentclass{article}
\usepackage[english]{babel}
\usepackage[letterpaper,top=2cm,bottom=2cm,left=3cm,right=3cm,marginparwidth=1.75cm]{geometry}
\usepackage{amsmath}
\usepackage{amsthm}
\usepackage{graphicx}
\usepackage{color}
\usepackage{bbm}
\usepackage{booktabs}
\usepackage{algorithm,algorithmic}
\usepackage{caption}
\usepackage{subfigure}
\usepackage{subcaption}
\newtheorem{theorem}{Theorem}[section]
\newtheorem{corollary}{Corollary}[theorem]
\newtheorem{lemma}[theorem]{Lemma}
\newtheorem{Proposition}{Proposition}[section]
\theoremstyle{definition}
\newtheorem{definition}{Definition}[section]
\newtheorem{remark}[theorem]{Remark}

\newcommand\keywords[1]{\textbf{Keywords}: #1}

\title{Multi-Agent Coverage Control in Non-Convex Annulus Region with Conformal Mapping}
\author{Xun Feng, Chao Zhai}

\author{Xun Feng, Chao Zhai \thanks{Xun Feng and Chao Zhai are with the School of Automation, China University of Geosciences, Wuhan 430074 China,
and with Hubei Key Laboratory of Advanced Control and Intelligent Automation for Complex Systems and Engineering Research Center of Intelligent Technology for Geo-exploration, Ministry of Education, Wuhan 430074 China. Corresponding author: Chao Zhai (email: zhaichao@amss.ac.cn).}
}

\begin{document}
\maketitle

\begin{abstract}
Efficiently fulfilling coverage tasks in non-convex regions has long been a significant challenge for multi-agent systems (MASs). By leveraging conformal mapping, this paper introduces a novel sectorial coverage formulation to transform a non-convex annulus region into a topologically equivalent one. This approach enables the deployment of MASs in a non-star-shaped region while optimizing coverage performance and achieving load balance among sub-regions. It provides a unique perspective on the partitioned sub-regions to highlight the geodesic convex property of the non-star-shaped region. By utilizing the sectorial partition mechanism and the diffeomorphism property of conformal mapping, a decentralized control law is designed to drive MASs towards a desired configuration, which not only optimizes the global coverage cost but also ensures exponential convergence of equitable workload.  Moreover, an iterative search algorithm is developed to identify the optimal approximation of multi-agent deployment in the non-star-shaped region. Theoretical analysis is conducted to confirm the asymptotic stability and global convergence with arbitrary small tolerance of the closed-loop system. Finally, numerical simulations demonstrate the practicality of the proposed coverage formulation with conformal mapping.
\end{abstract}
\keywords{Coverage Control, Multi-agent Systems, Sectorial Partition, Conformal Mapping}
\section{Introduction}
The advent of low-cost smart sensors and robots has facilitated extraordinary progress in multi-agent coordination, enabling the accomplishment of complex practical missions. Multi-agent systems offer the benefit of cooperative endeavors, facilitating the achievement of arduous tasks for a solitary agent to conduct. Such operations include, but are not limited to, area coverage~\cite {zhai13}, environment monitoring~\cite{zhai21}, autonomous cruise~\cite{vehicle23}, and border patrolling~\cite{patrol22}. The successful accomplishment of these intricate missions invariably necessitates the design of cooperative control strategies for multi-agent systems (MASs), which has increasingly captured the attention of researchers over the past few decades. In the applications mentioned above, multi-agent coverage control is an extremely important part 
due to its crucial roles in industry, agriculture, and the military, and it focuses on the problem of attaining desired coverage quality of an environment with metric considerations~\cite{Cortes04},~\cite{Cortes09}.

%In the context of centralized systems, the aggregation of information from all agents to a central node for processing incurs substantial computational costs and communication burdens on the central node, which can become prohibitive. Moreover, the system integrity is jeopardized if the control center fails, potentially leading to a breakdown of the entire system. In contrast, the decentralized system allows each agent to make decisions based on local information and limited communication with neighboring agents, which significantly alleviates the pressure of centralized processing and enables agents to continue working according to the locally stored information and preset rules in the event of a communication breakdown.  

Region partitions are conducive to the distributed coverage control of MASs by designating one subregion for each agent. 
The most commonly acknowledged analytic tools for distributed coverage control include the Voronoi partition and equitable workload partition. For instance, one common approach is to design the feedback control law that drives the agent to the centroid of Voronoi cell, which results from the Voronoi tessellation of the convex region~\cite{Cortes04}. To extend this approach beyond convex environments, geodesic Voronoi tessellation has been proposed~\cite{Pimenta08}, and it is determined by geodesic distance rather than Euclidean distance. However, this approach is restricted to Euclidean environments that contain polygonal obstacles. To solve the environmental limitation, a graph search-based approach is utilized for solving the coverage problem in non-convex environments with a non-Euclidean metric~\cite{Unknown13}, which suffers from instability in computing the generalized
centroid of Voronoi tessellation. To address this stability issue, the geodesic power Voronoi tessellation has been adopted in general non-convex Riemannian manifolds with boundaries, where the intrinsic metric is non-Euclidean~\cite{SBhattacharya14}. 
 In another direction, equitable partition policies are particularly attractive to researchers due to their broad applications and the novel insights they provide into the properties of power diagrams~\cite{equitable09}. A static coverage optimization problem with load balance has taken into account in~\cite{Cortes10}, where the control strategy is devised to optimize the coverage cost while moving MASs toward the generalized Voronoi centroid. However, Voronoi partitions might fail to maintain the connectivity of subregions when load balancing is required~\cite{{equitable19}}. To address the above issue, a distinct coverage formulation based on equitable workload partition has been proposed by dividing the coverage region into multiple stripes and further partitioning each stripe into sub-stripes with the same workload~\cite{zhai13}~\cite{zhai20}. Nevertheless, these studies focus on the coverage problem of regular regions (e.g. convex polygon and regions with parallel boundaries). For the star-shaped non-convex region, a sectorial coverage control approach for both load balancing and performance optimization has been developed~\cite{zhai23}, although its application is limited to star-shaped environments.

To overcome these limitations, thanks to the development of conformal mapping theory for simply-connected open surface meshes~\cite{computerconformalGeometry07, Teichmuller18}~\cite{rectangular13}~\cite{rectangular16}~\cite{diskconformal15} and multiply-connected surface meshes~\cite{multiconnected21}, this paper aims to map the general non-convex region to a topological equivalent one that is convenient to multi-agent coverage tasks. Inspired by the mapping in robot navigation~\cite{usingharmonicmaps18}, this paper presents a decentralized control strategy for workload balance and coverage performance optimization in non-star-shaped regions with the assistance of conformal mapping theory. In comparison with the previous non-convex coverage methods, such as greedy-based algorithm~\cite{sun19}, boosting function method~\cite{EscapeLocalOptima20}, energy-aware strategy for task allocation~\cite{Avoidance22}, performance bound computation in parallel with the greedy algorithm~\cite{sub22}, the proposed strategy enables the approximation of global optima with guaranteed load balance, collision avoidance among agents, and obstacle avoidance in non-star-shaped regions. 
While existing coverage strategies based on mapping theory map nonuniform metrics to a near Euclidean metric~\cite{Nonuniform08}, the control law is centralized and requires global information upfront. 
Previous work has also transformed connected regions into almost convex regions through diffeomorphism~\cite{perform08}, 
but it lacks the ability to describe the obstacle message and is unsuitable for complex obstacles. 
By considering the evolution of the real coverage domain, a control law was developed to capture the effects of time-varying diffeomorphism. However, it only addressed the problem of the rectangular region with rectangular obstacles~\cite{time-varying20}. 

To tackle with general scenarios, this paper proposes a coverage control approach via conformal mapping technology, 
which enables the handling of non-convex regions with a complex obstacle while ensuring workload equalization, global convergence, and subregion connectivity. In brief, the core contributions of this work are listed as follows: 
\begin{enumerate}
\item  Propose a theoretical formulation for distributed design of multi-agent coverage control in compact Riemannian manifold, which enables the coverage of non-convex irregular region with non-smooth boundary. 
\item  Construct an analytic diffeomorphism with the aid of conformal mapping theory to transform a non-convex irregular region into a topologically equivalent(i.e., star-shaped) region, which allows the distributed partition and coverage of non-convex environment. 
\item  Introduce a length metric for improving the distributed control law, which guarantees collision avoidance with obstacle or non-smooth boundary of non-convex environment. 
\end{enumerate}

% compared: drawback about automatica thesis, benefits of this method.

The remainder of this paper is structured as follows. Section~\ref{Preliminaries} introduces the preliminary concepts frequently used in this paper, including the topological properties of the workspace and details of the mapping design. Section~\ref{Coverage Control} formulates the coverage control problem of MASs in a non-star-shaped environment and presents a distributed coverage algorithm based on conformal mapping. Section~\ref{CaseStudy} presents numerical demonstrations to validate the coverage approach by comparing it with Voronoi tessellation and energy-aware task allocation and illustrating the validity of the sectorial coverage formulation. Section~\ref{Conclusions} concludes the paper and discusses future work.

\section{Preliminaries}\label{Preliminaries}
This section presents a selection of concepts related to Riemannian Geometry and conformal mapping that are most prevalently utilized in the context of this paper. For a further discussion about these topics, refer to~\cite{computerconformalGeometry07, Riemannian geometry}. Conformal mapping is an effective tool for establishing an equivalence relation between topologically equivalent surfaces. In the process of conformal mapping, the magnitude and direction of the angle can be kept strictly unchanged. Essentially, quasi-conformal mapping is the generalization of conformal mapping, and it allows a bounded angle distortion.
\begin{definition}[Conformal Map~\cite{computerconformalGeometry07}]
Let \(f:\mathbbm{C} \to \mathbbm{C}\) be a mapping with \(f\left( z \right) = f\left( {x,y} \right) = u\left( {x,y} \right) + iv(x,y)\), where \(u\), \(v\) are real-valued functions. \(f\) is said to be a conformal map and is differentiable everywhere in \(\mathbbm{C}\), for any \(z \in \mathbbm{C}\), \(f'\left( z \right) \ne 0\), and \(f\) preserves angle. In other words, \(f\) satisfies the Cauchy–Riemann equations
\begin{equation}  
\begin{aligned}
{\frac{{\partial u}}{{\partial x}} = \frac{{\partial v}}{{\partial y}}},\quad
{\frac{{\partial v}}{{\partial x}} =  -\frac{{\partial u}}{{\partial y}}}.\\
\end{aligned}
\end{equation}
\end{definition}

\begin{definition}[Quasi-Conformal Map~\cite{computerconformalGeometry07}]
A mapping \(f:\mathbbm{C} \to \mathbbm{C}\) is said to be a quasi-conformal map if it satisfies the Beltrami equation
\begin{equation}  
\frac{{\partial f}}{{\partial \bar z}} = {\mu _f}\left( z \right)\frac{{\partial f}}{{\partial z}}\
\end{equation}
for some complex-valued function \({\mu _f}\) with \({\left\| {{\mu _f}} \right\|_\infty } < 1\), where the complex derivatives are given by
\begin{equation}  
\left\{
\begin{aligned}
{\frac{{\partial f}}{{\partial z}} = \frac{1}{2}\left( {\frac{{\partial f}}{{\partial x}} -i\frac{{\partial f}}{{\partial y}}} \right)} \\ 
  {\frac{{\partial f}}{{\partial \bar z}} = \frac{1}{2}\left( {\frac{{\partial f}}{{\partial x}} + i\frac{{\partial f}}{{\partial y}}} \right)} \\
\end{aligned}
\right.
\label{equation3}
\end{equation}
where \({\mu _f}\) is called the \textit{Beltrami coefficient} of \(f\). 
\end{definition}
Next, we introduce the concept of the coordinate chart that allows to transform a general manifold into the Euclidean space.
\label{Def2.6}

\begin{definition}[Coordinate Chart~\cite{Topology99}]
A topological space \(M\) is locally Euclidean of dimension \(n\) if every point \(p\) in \(M\) has a neighborhood \(U\) such that there is a homomorphism \(\phi \) from \(U\) onto an open subset of      \({\mathbbm{R}^n}\). We call the pair \(\left( {U,\phi :U \to {\mathbbm{R}^n}} \right)\) a coordinate chart.
\end{definition}
Subsequently, we introduce the concept of Riemannian metric, which endows the tangent space
with an inner product structure. This metric can induce practical metric structures suitable for multi-agent systems to perform coverage tasks with collision avoidance.
\begin{definition}[Riemannian Metric~\cite{Riemannian geometry}]
 A Riemannian metric \(\eta\) on \(M\) is an assignment of an inner product
 \({\eta_p}\left( { \cdot , \cdot } \right) = {\left\langle { \cdot , \cdot } \right\rangle _p}\) on      \({T_p}M\) (i.e. tangent space at a point $p$) for each \(p \in M\) that depends smoothly on $p$. It is noted that \(\eta\) itself is not a metric (actually a distance function) on \(M\).
\end{definition}

Although the cut locus is a troublesome set that contributes to the non-uniqueness of the optimal path, we suppose that it is a Lebesgue-measure-zero set in the coverage region. This point will be discussed in the later section.
\begin{definition}[Cut Locus~\cite{Cut loci85}]
 Let  \(p \in \Omega\) and \(q \in \left( {\Omega  - \partial \Omega } \right)\), we call \(q\) is a    \(pica\) relative to \(p\) if there are more than one minimal paths that can connect \(p\) and \(q\), the closure of the set of all \(pica\) relative to \(p\) is called the cut locus of \(p\), which is denoted as \({C_p}\).
\end{definition}
There exists a unique minimum geodesic connecting any two points in a geodesically convex region. The key technical challenge of this paper is to establish the geodesically convex nature of the coverage region, which is essential to analyze the existence of optimal path.
\begin{definition}[Geodesically Convex Region~\cite{Riemannian geometry}]
Suppose \(\left( {\Omega,{\eta}} \right)\) is a connected Riemannian manifold with Riemannian metric \(\eta\), and there is a unique minimum geodesic connecting \(p\) and \(q\), which is entirely within this region. We call \(\Omega\) a geodesically convex region.
\end{definition}

\begin{figure}[t!]
\centering
\includegraphics[width=0.3\textwidth]{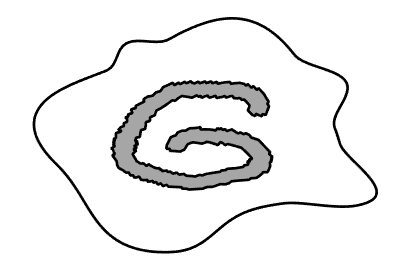}
\caption{\label{png: MappingRegionRaw} Illustration of non-convex coverage region \(\Omega\), whose obstacle has non-smooth boundary. It is not easy to directly deploy the multi-agent system in this region because of the absence of a common reference point for the inner and outer star-shaped sets, which implies that it is infeasible to realize the sectorial partition.}
\label{Fig.1.}
\end{figure}

\subsection{Workspace Description}
Consider a 2-dimensional manifold \(\Omega\), which is enclosed by a closed curve (see Fig.~\ref{Fig.1.}). The gray region represents the obstacle or hole, denoted as \(O\). It is noteworthy that the obstacle boundary is non-smooth. The region $\Omega-O$ presents a significant challenge for multi-agent deployment due to its non-smooth boundary and non-convex nature. The classic control law could drive agents to the boundaries or cause them to get stuck in the obstacle~\cite{Cortes04}. 
%Pimenta et al.\cite{Pimenta08} propose using geodesics to measure the distance between start and target points, but this approach encounters a problem where the geodesic may traverse the boundary of \(\Omega-O\). 
To effectively navigate and avoid obstacles, we need to consider the boundary of the obstacle \(O\). The region \(\left({\Omega-O}\right)\cup\partial O\) forms a compact and non-convex region, which is also a complete metric space. According to the Hopf-Rinow Theorem~\cite{Riemannian geometry}, the existence of geodesics can be guaranteed in a complete metric space. For any two points within \(\left( {\Omega  - O} \right) \cup \partial O\), there exist geodesics that connect them shortest. To avoid collisions, we need to establish a suitable metric to measure the distance between the agent and the target. Although we cannot directly implement the MAS within the region, it is homeomorphic to an annular region, which implies that it has the same topological properties as a 2-dimensional unit annular region. Furthermore, we can employ an analytic diffeomorphism to transform the region into a unit annular region, which is conducive to controlling the MAS. We design the control law and performance index within the mapped region. Given the bijective of diffeomorphism, the original parameter representations (i.e. control input, performance index) of \(\left( {\Omega  - O} \right) \cup \partial O\) could be obtained.

\subsection{Design of Conformal Mapping}
\label{Mappingdesign}
In this subsection, we focus on designing an analytic diffeomorphism for the region  (\(\Omega\)-\(O\))\(\cup\)\(\partial O\). Our goal is to establish a one-to-one correspondence between this region and a star-shaped region where agents can be effectively deployed for distributed coverage. Mapping theory is introduced to demonstrate the feasibility of conformal mapping. For closed surfaces of genus one, a flat annulus with a conformally equivalent flatness metric can be obtained. The universal cover surface of this flat annulus isometrically and parsimoniously covers the entire Euclidean plane~\cite{computerconformalGeometry07}. However, this mapped surface is not suitable for multi-agent deployment. Hodge theory is used to construct the conformal mapping of a topological annulus~\cite{computerconformalGeometry07}, but the bijective property and stability of the conformal mapping are not guaranteed for overly complex initial surfaces. 
Nevertheless, conformal parameterization methods for simply connected surfaces have been extensively developed~\cite{Teichmuller18}~\cite{rectangular13}~\cite{rectangular16}~\cite{diskconformal15}. Therefore, we aim to transform the topological annulus into some simply connected surfaces. Thanks to the fact that each agent has its own perceptual region, which is simply connected, it naturally forms the simply connected partition sub-regions. 

A Riemannian surface \(S \) is considered simply connected if every closed curve within \(S \) is homotopic to a point. For the simply connected surface, there are already theoretical tools to analyze their topological structure. The uniformization theorem is a generalization of Riemannian Mapping Theorem~\cite{Riemannian geometry}, which significantly simplifies the study of simply connected Riemannian surfaces by reducing it to the study of the disk, plane, and sphere.  
\begin{lemma}[Poincare-Koebe Uniformization~\cite{computerconformalGeometry07}]
Every simply connected Riemannian surface is conformally equivalent to the unit disk, complex plane, or Riemannian sphere.
\label{Poincare}
\end{lemma}

Consider the topological annulus (\(\Omega\)-\(O\))\(\cup\)\(\partial O\) with its outer boundary as \(\partial\Omega\) and inner boundary denoted as \(\partial O \). Suppose there are \(N\) agents, with \(v^i\) representing the set of vertices under the supervision of the \(i\)-th agent, implying that each agent possesses a point cloud set \(v^i\). Then we denote the \(i\)-th sub-region as \(V_i\), which consists of the point cloud set \(v^i\). The Delaunay triangulation technique is employed to point cloud set \(v^i\), which signifies that the \(i\)-th sub-region is approximated by a simplicial complex, with Delaunay triangulation serving as a cell decomposition. The \(i\)-th agent is equipped with a diffeomorphism \(\vartheta_i=u_i+jv_i\), where \(\vartheta _i\) is a rectangular conformal mapping~\cite{rectangular16} with \(j=\sqrt{-1}\). \(\vartheta _i\) uses disk harmonic mapping \(H_i\) as the base mapping. It means that the \(i\)-th sub-region can be mapped to a unit disk \(\mathbbm{D}\), then the unit disk \(\mathbbm{D}\) be mapped to a rectangular region \([0, L_i] \times [0,1]\). 

\begin{figure}[t!]
\centering
\includegraphics[width=1\textwidth]{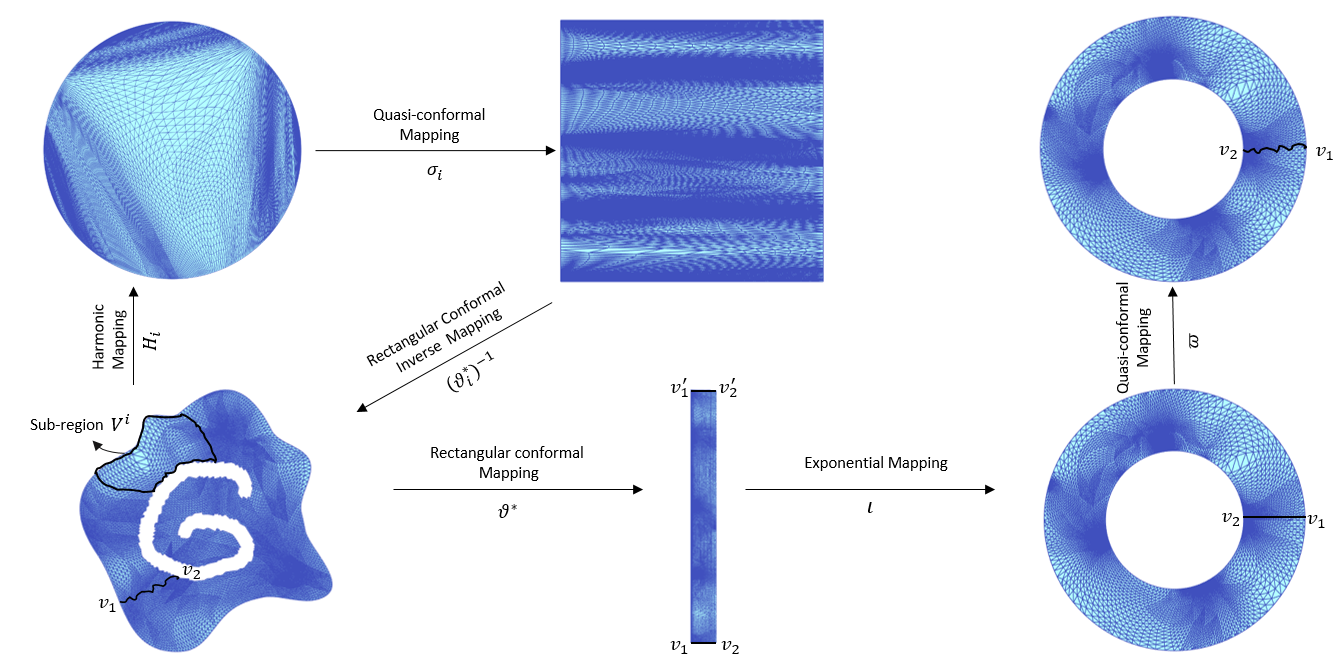}
\caption{\label{png: MAPProcess1} An illustration of the conformal mapping process. At first, we use the harmonic mapping \(H_i\) for the sub-region \(V^i\), establishing a harmonic diffeomorphism between \(V^i\) and unit disk \(\mathbbm{D}\). Using \(\sigma_i\) to transform the unit disk \(\mathbbm{D}\) to the rectangular region \([0, L_i] \times [0,1]\). Through Algorithm~\ref{Algorithm1}, we obtain the optimum length \(L_*\). Then, the rectangular conformal inverse mapping \((\vartheta _i^*)^{-1}\) maps the region \([0, L_i] \times [0,1]\) to point cloud \(V^i_*\). According to the corresponding mechanism among agents and the ICP point cloud registration, we obtain all the coordinates in (\(\Omega\)-\(O\))\(\cup\)\(\partial O\). Slicing the mesh (\(\Omega\)-\(O\))\(\cup\)\(\partial O\) along a path \(v_1v_2\) yields a simply connected open region \(\tilde \Omega \). Secondly, we construct a rectangular conformal mapping \(\vartheta_i^*\) between \(\tilde \Omega \) and the rectangular region \([0, L_*] \times [0,1]\) with an optimal length \(L_*\) and width 1. The rectangle is subsequently mapped to an annulus using an exponential map \(\iota  ={e^{2\pi \left( {z -L_*} \right)}}\). Finally, we identify the cut vertices \(v_1v_2\) and compose the quasi-conformal map \(\varpi \) to reduce angle distortion, which forms a conformal parameterization. The last region is \(\Xi\). }
\label{Fig.2.}
\end{figure}

According to Lemma \ref{Poincare}, we can construct a harmonic diffeomorphism between \(V_i\) and the unit disk \(\mathbbm{D}\) (see Fig.~\ref{png: MAPProcess1}). Although harmonic mapping is generally not conformal, it is suitable for this purpose due to its diffeomorphism property and the relatively small angle distortion. For the \(i\)-th agent in sub-region \(V_i\), compute a disk harmonic mapping \({H_i}: V_i  \to \mathbbm{D}\) by solving the following Laplace equation 
\begin{equation}
\Delta {H_i} = 0,\quad {H_i}\left({\partial V_i}\right)=\partial\mathbbm{D},
\end{equation}
where \(\Delta\) is the Laplace–Beltrami operator on \({ V_i}\), and it can be easily discretized using the cotangent Laplacian~\cite{rectangular16}.  This discretization transforms the Laplace equation \(\Delta {H_i} = 0\) into a sparse, symmetric positive definite linear system, which allows us to readily obtain the solution. The Beltrami coefficient \(\mu \left( {{H _i}} \right) = {\lambda _i} + j{\gamma _i}\) measures the conformality distortion of a map \(H_i\). In fact, a smaller norm of the Beltrami coefficient implies a tinier conformality distortion. With the Dirichlet boundary conditions~\cite{rectangular13}, the problem of solving the disk harmonic mapping can be converted to solve the partial differential equations 
\[\nabla  \cdot \left( {A\left( \begin{array}{l}
{\left( {{u_i}} \right)_x}\\
{\left( {{u_i}} \right)_y}
\end{array} \right)} \right) = 0,
{\rm{ }}\nabla  \cdot \left( {A\left( \begin{array}{l}
{\left( {{v_i}} \right)_x}\\
{\left( {{v_i}} \right)_y}
\end{array} \right)} \right) = 0,
{\rm{ }}A = \left( {\begin{array}{*{20}{c}}
{{\alpha _1}}&{{\alpha _2}}\\
{{\alpha _2}}&{{\alpha _3}}
\end{array}} \right),\]
where 
$${\alpha _1} = \frac{{{{\left( {{\lambda _i} - 1} \right)}^2} + {{\left( {{\gamma _i}} \right)}^2}}}{{1 - {{\left( {{\lambda _i}} \right)}^2} - {{\left( {{\gamma _i}} \right)}^2}}}, \quad {\alpha _2} =  - \frac{{2{\gamma _i}}}{{1 - {{\left( {{\lambda _i}} \right)}^2} - {{\left( {{\gamma _i}} \right)}^2}}},\quad {\alpha _3} = \frac{{{{\left( {{\lambda _i} + 1} \right)}^2} + {{\left( {{\gamma _i}} \right)}^2}}}{{1 - {{\left( {{\lambda _i}} \right)}^2} - {{\left( {{\gamma _i}} \right)}^2}}}.
$$ 
By solving the above equations, we obtain the \(x\)-coordinate and \(y\)-coordinate functions of \(H_i\), which are used to obtain the coordinates of the point cloud set $v^i$ in unit disk \(\mathbbm{D}\). In practice, quasi-conformal mappings are used to reduce the angle distortion. The Linear Beltrami Solver (LBS) method~\cite{rectangular13} with the Beltrami coefficient \({\mu _{{\sigma _i}}} = {\mu _{{H^{-1}_i}}}\) is employed to compute a quasi-conformal map \({\sigma _i}:\mathbbm{D} \to [0, L_i] \times [0,1]\) from the unit disk to a rectangular domain with the length \(L_i\) and the width 1, where \( L_i\) minimizes the norm of the Beltrami coefficient of \({\sigma _i} \circ {H_i}\) to reduce the conformal distortion. In the computation process of rectangular conformal mapping, the \(x\)-coordinate and \(y\)-coordinate of \(i\)-th agent are updated by 
\[\left\{ \begin{array}{l}
x_{new} = {L_i}x_{old}\\
y_{new} = y_{old}
\end{array} \right.\]
Then, \(x_{new}\) and \(y_{new}\) are used to obtain the norm of the Beltrami coefficient of \({\vartheta _i} \circ {H_i}\), implying the length \(L_i\) affects the norm of Beltrami coefficient. To relieve the conformal distortion, we need to minimize the norm of the Beltrami coefficient. Thus, the optimal index is proposed as follows:
\begin{equation}
{L_*} =\arg \min \left\| \mu  \right\|
\label{L^*}
\end{equation}
where \(\mu  = \left( {{\mu _1}, \cdots {\mu _N}} \right)\), the optimum length \(L_*\) minimizes the norm of the Beltrami coefficient \(\left\| \mu \right\|\), indicating the minimal conformal distortion. Under the topology of uniform convergence on the region, the function space composed of rectangular conformal maps is compact~\cite{computerconformalGeometry07}. Due to the function space being compact, the existence of optimal index~\eqref{L^*} is guaranteed.
\label{Compact Function Space}

\begin{remark}
By Definition~\ref{Def2.6}, the rectangular conformal mapping \(\sigma_i\) satisfies \({\left\| {{\mu _{\sigma_i} }} \right\|_\infty } < 1\), it ensures the uniform boundedness of \(\sigma_i\). Since the mapping \(\sigma_i\) is defined on a compact region \(\mathbbm{D}\), \(\sigma_i\) is equicontinuous. According to Arzelà–Ascoli Lemma~\cite{computerconformalGeometry07}, the function space composed of rectangular conformal maps \(\sigma_i\) is compact.
\end{remark}

Through the certain coordination mechanism among agents (see details in Algorithm~\ref{Algorithm1}), the agents can share the local optimal length with their neighborhoods. Because the function space composed of rectangular conformal maps \(\sigma_i\) is compact, we can obtain the global optimal length \(L_*\). We denote the optimal mapping corresponding to the optimal length \(L_*\) as \(\vartheta^*\). For the \(i\)-th agent in the rectangular region \([0, L_i] \times [0,1]\), after we obtain the optimum length \(L_*\), the vertices in set $v^i$ is updated 
by 
\[\left\{ \begin{array}{l}
x_{new}' = \frac{{{L_*}}}{{{L_i}}}x_{new}\\
y_{new}' = y_{new}
\end{array} \right.\]
Through this computation and the coordination mechanism among agents, all agents can be located in the optimum rectangular \([0, L_*] \times [0,1]\). Since rectangular conformal mapping \(\vartheta^*\) is a diffeomorphism, we can use mapping \((\vartheta^*)^{-1}\) to map the rectangular region \([0, L_i] \times [0,1]\) to a new point cloud set \(V^i_*\) in the original region (\(\Omega\)-\(O\))\(\cup\)\(\partial O\). 

Following the design of conformal mapping, our attention is turned to collecting environmental information in a distributed manner. Each agent is only aware of information pertaining to the region it oversees and is able to share information with its neighbors. The communication process can be viewed as point cloud registration~\cite{computerconformalGeometry07}. In the mapping design process, we only need to use the coordinate locations of the point cloud. Thus, the classical registration method Iterative Closest Point (ICP)~\cite{Pointcloud} could be employed to achieve the point cloud registration in (\(\Omega\)-\(O\))\(\cup\)\(\partial O\). Subsequently, we obtain the global point cloud \({V_*} = \left( {{V^{1}_*}, \cdots ,{V^{N}_*}} \right)\) in (\(\Omega\)-\(O\))\(\cup\)\(\partial O\). The Delaunay triangulation technique is employed to the global point cloud \({V_*} = \left( {{V^{1}_*}, \cdots,{V^{N}_*}} \right)\), which means that (\(\Omega\)-\(O\))\(\cup\)\(\partial O\) is approximated by a simplicial complex. This approach ensures that the definition of tangent vector length and the included angle remain consistent between adjacent triangles, underscoring the continuity of tangent vector length change. 

We select an arbitrary vertex \(v_1\) on the inner boundary \(\partial O \), identify the closest vertex \(v_2\) on the outer boundary \(\partial \Omega \), and choose the shortest path between these vertices. This process yields a simply connected open region \(\tilde\Omega\), with a new inner point \(v'_1\) and outer point \(v'_2\). Then we use the rectangular conformal mapping \(\vartheta^*\) to map the region \(\tilde\Omega\) to the region \([0, L_*] \times [0,1]\), which satisfies \({\vartheta ^*}\left( {{v_1}} \right) = \left( {0,0} \right)\), \({\vartheta ^*}\left( {{v_2}} \right) = \left( {L_*,0} \right)\), \({\vartheta ^*}\left( {{v'_1}} \right) = \left( {0,1} \right)\), \({\vartheta ^*}\left( {{v_2'}} \right) = \left( {L_*,1} \right)\). Mapping to a rectangular domain helps to simplify the concatenation of subsequent annulus surfaces. The exponential map \({\iota} = {e^{2\pi \left( {z - L_*} \right)}}\) enables to convert the rectangular domain \([0, L_*] \times [0,1]\) to a new annulus \({\Omega '}\) with inner radius \({e^{ - 2\pi L_*}}\) and outer radius 1. The cut path \({v_1}{v'_1}\) and \({v_2}{v'_2}\) are mapped to consistent locations on  \({\Omega '}\) because of the periodicity imposed in the rectangular parameterization~\cite{multiconnected21}. Finally, a quasi-conformal mapping \({\varpi}:\Omega ' \to \Xi \)  is used to obtain an automorphism \(\varpi\) of the topological annulus that is suitable for multi-agent deployment~\cite{zhai23}. The automorphism is obtained by the Linear Beltrami Solver (LBS) as follows:
\[{\varpi} = \text{LBS}\left( {{\mu _{{{\left( {{\iota} \circ {\vartheta ^*} } \right)}^{ - 1}}}}}\right),\]
where \({\mu _{{{\left( {{\iota} \circ {\vartheta ^*} } \right)}^{ - 1}}}}\) is the Beltrami coefficient of the composite mapping \({{{\left( {{\iota} \circ {\vartheta ^*} } \right)}^{ - 1}}}\).
The calculation of the Beltrami coefficient for a composite function and the associated angle distortion is detailed in~\cite{diskconformal15}. The conformal mapping of an annulus can be expressed as a composite mapping \(\tau = \varpi \circ \iota \circ \vartheta ^*\), where \(\tau: (\Omega - O) \cup \partial O \to \Xi\). The quasi-conformal mapping ensures that \(\tau\) is a conformal mapping with arbitrarily small angle distortion. The specific steps are illustrated in Fig.~\ref{png: MAPProcess1}. 
In the complex plane \(\mathbbm{C}\), note that \(f\left( z \right) = f\left( {x,y} \right) = u\left( {x,y} \right) + iv(x,y)\) with \(u\), \(v\) are real-valued functions, and the Jacobian \({J_f}\) of quasi-conformal mapping \(f\) is given by
\begin{equation}
{J_f} = {\left| {{f_z}} \right|^2}\left( {1 - {{\left| {{\mu _f}} \right|}^2}} \right),
\label{Jacobian of f}
\end{equation}
where \({\mu _f}\) is the Beltrami coefficient of \(f\). This concept can be naturally extended for the  Riemannian manifold with the aid of local charts. 
\begin{remark}
According to Definition~\ref{Def2.6}, the mapping \(\varpi\) satisfies \({\left\| {{\mu _\varpi }} \right\|_\infty } < 1\). Based on Equation \eqref{Jacobian of f}, the Jacobian determinant of the quasi-conformal mapping \(\det(J_\varpi)\) is positive. Owing to the Jacobian determinant of conformal mapping being positive, we obtain \(\det(J_{\vartheta^*})>0\). Since \(\det(J_\iota)>0\) and  \(\tau = \varpi \circ \iota \circ \vartheta ^*\), we have \(\det \left( {{J_\tau }} \right) > 0\), indicating that \(\tau\) is an orientation-preserving diffeomorphism. 
\end{remark}
 
\begin{remark}\label{Estimation}
In the process of mapping design, Delaunay triangulation is used to discretize the region (\(\Omega\)-\(O\))\(\cup\)\(\partial O\). When each agent is driven to its optimal location, it happens that the optimal point is inside the Delaunay triangle, which means the point is not at the point cloud set. Hence, this leads to the discretization error.
The discretization error of the Delaunay triangulation can be estimated in the mapping process.
Denote the actual solution and the optimal solution as \({v_{a}}\) and \({v_{o}}\), respectively. In the 2-dimensional compact manifold (\(\Omega\)-\(O\))\(\cup\)\(\partial O\), we set point \({v_{o}} = ({x_{o}},{y_{o}})\). Without loss of generality, we suppose that the maximum difference between \({v_{a}}\) and \({v_{o}}\) is \(\left\| {{x_{a}} - {x_{o}}} \right\|\). Suppose the numerical difference \({\left\| {{x_a} - {x_o}} \right\|_\infty } \le h <1\), where \(h\) is the meshing size and \(p\) is a convergence index with \(p>1\). It satisfies that \(\forall \varepsilon  > 0\), \(\exists \delta  = {\left( {\varepsilon /\sqrt n } \right)^{\frac{1}{p}}}\), \({\left\| {{x_a} - {x_o}} \right\|_\infty } < \delta \), \({\left\| {{v_a} - {v_o}} \right\|_\infty } \le {\left\| {{v_a} - {v_o}} \right\|_2} \le \sqrt n {\left\| {{v_a} - {v_o}} \right\|_\infty } \le \sqrt n {\left\| {{x_a} - {x_o}} \right\|_\infty } < \sqrt n {\left\| {{x_a} - {x_o}} \right\|_\infty^p } \le \sqrt n \delta  = \varepsilon \). Then we obtain \({\left\| {{v_a} - {v_o}} \right\|_\infty } < \sqrt n {h^p}\). %According to the error convergence analysis in~\cite{mesherror}, the discretization error of the Delaunay triangulation can be estimated in this mapping process.
\end{remark}

\section{Coverage Control in Non-Star Shaped Region}
\label{Coverage Control}
The workspace under consideration is a general Riemannian manifold that is not complete but can be regarded as a complete metric space. With the above conformal mapping, we can transform it into a unit annular region, and Lemma~\ref{Poincare} guarantees the point-to-point correspondence. In this section, we first describe the coverage control problem in the mapped workspace. Then, we design the control law in the original workspace by utilizing the bijective nature of the conformal mapping. Finally, we illustrate the validity of the proposed sectorial coverage algorithm.

\subsection{Design of Coverage Controller}
This subsection formulates the coverage problem of multi-agent systems for optimal monitoring in irregular non-convex environment with non-smooth boundary (see Fig.~\ref{Fig.2.}) which is embedded in \(\mathbbm{R}^2\). The mathematical
expression of the mapped coverage region \(\Xi \) is given by
\begin{equation}  
\left\{
\begin{aligned}
  {x(\varphi,\theta) = \left( {R + r\cos \varphi } \right)\cos \theta } \\ 
  {y(\varphi,\theta) = \left( {R + r\cos \varphi } \right)\sin \theta } \\
\end{aligned}
\right.
\label{coordinate representative}
\end{equation}
with \(\theta,\varphi \in [0,2\pi )\), \(R = ({{1 + {e^{ - 2\pi L_*}}}})/{2}\) and \(r = ({{1 - {e^{ - 2\pi L_*}}}})/{2}\), where \(L_*\) is obtained during the communication process of rectangular conformal maps among agents. The goal is to design a control law that ensures the agents effectively cover the entire \(\Xi\) through a sectorial partition while achieving workload balance.

Let \(\Xi \) be a compact region, and consider a distribution density function \(\rho:\Xi  \to {\mathbbm{R}^+}\) which quantifies the amount of information or workload at any point within \(\Xi\). We employ the sectorial partitioning technique to divide the region partition and balance the workload. Each agent is equipped with a virtual partition bar, represented by \(\{ \left( {\varphi,{\psi_i}} \right) \subset \Xi |i\in {I_N} = \{ 1,2, \ldots, N\} \} \), where \({\psi_i} \in [0,2\pi )\) is the phase angle of the partition bar for the \(i\)-th agent. The partition bars are numbered sequentially as \({\psi_1}\left( 0 \right) < {\psi _2}\left( 0 \right) <  \cdots  < {\psi _N}\left( 0 \right) < 2\pi \), with 
\({\psi _i}\left( 0 \right)\) denoting the initial phase of the partition bar. Consequently, a group of \(N\) agents enable to divide off the coverage region \(\Xi\) into \(N\) sub-regions (i.e. \(\Xi  = \bigcup {_{i = 1}^N} {E_i}\)), where \({E_i}\) represents the \(i\)-th subregion covered by the \(i\)-th agent, and it is enclosed by inner and outer curves as well as the partition bars \(\psi_i\) and \(\psi_{i+1}\).
The area in charge of the $i$-th agent is given by
\begin{equation}  
{E_i}\left( \psi  \right) = \left\{ \begin{aligned}
  \{ \left( {\varphi ,\theta } \right) \in \Xi |{\psi _i} < \theta  < 2\pi ,0 < \theta  < {\psi _{i + 1}}\} ,&~\text{if}~{\psi _{i + 1}} < {\psi _i} \hfill \\
  \{ (\varphi ,\theta ) \in \Xi |{\psi _i} < {\psi _{i + 1}}\} .&~\text{otherwise} \hfill \\ 
\end{aligned}  \right.
\label{Partitionarea}
\end{equation}
Subsequently, we set \(L = R + r\cos \varphi \), where the function \(\omega \left( {\theta } \right)\) is given by  
$$
\omega \left( \theta  \right) = \int_{{L_{in}}\left( \varphi  \right)}^{{L_{out}}\left( \varphi  \right)} {\rho (\theta ,L)}LdL.
$$
The workload on the  \(i\)-th subregion can be calculated by 
\begin{equation}  
{m_i} = \left\{ 
\begin{aligned}
\int_{{\psi_i}}^{2\pi} {\omega(\theta)d\theta}  + \int_0^{{\psi_{i+1}}}\omega(\theta)d\theta, 
&~\text{if}~{\psi_{i+1}} <{\psi_i}   \hfill \\
\int_{{\psi_i}}^{{\psi_{i+1}}}\omega(\theta)d\theta. 
&~\text{otherwise}\hfill \\ 
\end{aligned} \right.
\end{equation}
with \({\psi _{N + 1}} = {\psi _1}\) and \({m_N} = {m_0}\). To balance the workload among subregions, the dynamics of the partition bar are designed as
\begin{equation}  
{\dot \psi _i} = {k_\psi }\left( {{m_i} - {m_{i - 1}}} \right), \quad \forall i \in I_N
\label{partitioning dynamic}
\end{equation}
where \(k_\psi\) is a positive constant. To ensure that \({\psi _i} \in [0,2\pi )\),  the modulo operation is employed with \({\psi _i} = \bmod \left( {{\psi _i},2\pi } \right)\).
The dynamics of each agent are given by
\begin{equation}  
 {\dot p_i} = {u_i}, \quad \forall i\in {I_N}
 \label{PIandUI}
\end{equation}
with the control input \(u_i\).
After designing the partition dynamics in \(\Xi\), we proceed to demonstrate the existence of partition dynamics in \(({\Omega-O} )\cup \partial O\), where the diffeomorphism \(\tau\) is designed in Subsection~\ref{Mappingdesign}  and the Delaunay triangulation technique are employed.
\begin{theorem}
Delaunay triangulation and diffeomorphism \(\tau\) guarantee the existence of partition bars in the original space (\(\Omega\)-\(O\))\(\cup\)\(\partial O\).
\label{Theorem 3.1}
\end{theorem}

\begin{proof}
The triangular mesh structure generated by Delaunay triangulation is preserved under the diffeomorphism \(\tau\), indicating that the connectivity of regions remains unchanged. This implies that \(\tau\) maintains the adjacency of triangles formed by Delaunay triangulation in \((\Omega - O) \cup \partial O\), ensuring they remain adjacent in \(\Xi\). (\(\Omega\)-\(O\))\(\cup\)\(\partial O\) is topologically equivalent to an annular region and diffeomorphism \(\tau\) preserves the topology of the surface. Consequently, the Delaunay triangulation of the boundary is maintained on the boundary after the mapping \(\tau\), which guarantees that boundary points in \((\Omega - O) \cup \partial O\) are mapped to the boundary of \(\Xi\). The partition bars are continuous curves that connect the inner and outer boundaries in \(\Xi\). Due to the continuity of \(\tau^{-1}\) and the correspondence relationship of boundaries, the existence of partition bars is guaranteed in \((\Omega-O)\cup\partial O\).
\end{proof}

The coordinates in the original workspace are represented by \(\left( {\theta ',\varphi '} \right)\). It is noteworthy that partition bars are continuous curves, and \(\tau\) is a diffeomorphism that preserves continuity.  Similarly, as described in \eqref{Partitionarea}, the partition bar is represented by \({\psi '_i}\), and we obtain the region in charge of the $i$-th agent in (\(\Omega\)-\(O\))\(\cup\)\(\partial O\) as follows:
\begin{equation}  
{E'_i}\left( \psi  \right) = \left\{ \begin{aligned}
\{ \left( {\varphi ',\theta '} \right) \in (\Omega  - O) \cup \partial O|{\psi '_i} < \theta ' < 2\pi ,0 < \theta ' < {\psi '_{i + 1}}\} ,&~\text{if}~{\psi '_{i + 1}} < {\psi '_i}\\
\{ (\varphi ',\theta ') \in (\Omega  - O) \cup \partial O|{\psi '_i} < {\psi '_{i + 1}}\} .&~\text{otherwise}
\end{aligned} \right.
\end{equation}
Since the region (\(\Omega\)-\(O\))\(\cup\)\(\partial O\) is highly irregular and non-convex, it is a challenge to describe its area using coordinate parameters. To address this issue, we compute the areas of (\(\Omega\)-\(O\))\(\cup\)\(\partial O\) and \(\Xi\) by summing the areas of all Delaunay triangles, as Delaunay triangulation provides an accurate cell decomposition. We normalize the original region so that these areas are of the same order of magnitude. The distribution density function \(\rho:\Xi\to {\mathbbm{R}^+}\) is continuous. Since \({\tau }\) is continuous, the workload density function \(\rho ':(\Omega  - O) \cup \partial O \to {R^ + }\) is also continuous because the composition of a continuous function is still a continuous function.  We represent  the density  as follows: 
\begin{equation}  
\rho ' = \rho\circ {\tau }  \
\label{RhoOrigin}
\end{equation}
After normalizing the original region to ensure that the areas are of the same order of magnitude, the density change of the subregion depends on the Jacobian of \(\tau\). We calculate \(J_{\tau}\) using \eqref{Jacobian of f}, the density change is far less than the density itself during the mapping process because the quasi-conformal mapping \(\varpi \) can reduce the angle distortion which in turn contributes to reducing the density change. 

\begin{Proposition}
Conformal diffeomorphism \(\tau\) preserves the partition angle of each subregion.
\label{Proposition3.1}
\end{Proposition}
 
\begin{proof}
We analyze the impact of \(\tau\) on the partition bars. According to Theorem~\ref{Theorem 3.1}, the Delaunay triangulation guarantees the continuity of the length change of tangent vectors between adjacent triangles and the connectivity of regions remains unchanged. As \(\tau\) is an orientation-preserving diffeomorphism, it does not alter the local relative angle sequence (i.e.  \({\psi'_1}\left( 0 \right) < {\psi' _2}\left( 0 \right) <  \cdots  < {\psi' _N}\left( 0 \right) \) if \({\psi_1}\left( 0 \right) < {\psi _2}\left( 0 \right) <  \cdots  < {\psi _N}\left( 0 \right) \)). Moreover, it preserves the local distribution of Delaunay triangles. All of Delaunay triangles in chart \(\left( {{U_\alpha },{\phi _\alpha }} \right)\) will be mapped to \(\left( {\tau \left( {{U_\alpha }} \right),{{\phi'}_\alpha }} \right)\) where \({U_\alpha } \subset \Xi \), \(\tau \left( {{U_\alpha }} \right) \subset (\Omega  - O) \cup \partial O\). Thus, the relative positions of partition bars can be maintained in \((\Omega-O)\cup\partial O\) and \(\Xi\). Because the conformal diffeomorphism \(\tau\) preserves the angle of the Delaunay triangle and the one-to-one correspondence relationship of Delaunay triangle in boundary, the partition bar satisfies \({\psi _i} = {\psi' _i}\).
\end{proof}

%\begin{definition}[Workload in non star-shaped space]
In the mapped space, there exists a one-to-one correspondence between points and their polar coordinates, which facilitates the transformation of \({\rho (\theta, L)}\) into \({\rho (\theta,\varphi)}\). This approach enables an accurate representation of the serpentine obstacle boundaries. We derive the following expression:
$$
\omega \left( {\theta  } \right) = \int_{\varphi_{in}}^{\varphi _{out}} {\rho (\theta ,\varphi )} {\left( {R + r\cos \varphi } \right)^2}\sin \varphi d\varphi 
$$ 
with \(\theta ,\varphi \in [0,2\pi )\), \(R = ({{1 + {e^{ - 2\pi L_*}}}})/{2}\) and \(r = ({{1 - {e^{ - 2\pi L_*}}}})/{2}\), where \(L_*\) is obtained during the process of rectangular conformal mapping. Thus, we define the subregion workload in (\(\Omega\)-\(O\))\(\cup\)\(\partial O\) using \eqref{RhoOrigin}. In the original space  \((\Omega  - O) \cup \partial O\), we obtain
$$
\omega' \left( {\theta'} \right) = 
\int_{{\varphi '_{in}}}^{{\varphi'_{out}}} {{\tau ^{ - 1}}\left( {\rho (\theta ,\varphi )} \right)} {\left( {R' + r'\cos \varphi '} \right)^2}\sin \varphi'd\varphi'
$$ 
where \(R',r'\) can be represented by the preimage of \(R,r\). The workload on the  \(i\)-th subregion can be represented by 
\begin{equation}  
{m'_i} = \left\{ \begin{aligned}
\int_{{\psi'_i}}^{2\pi } {\omega '\left( {\theta' } \right)d\theta' }  + {\int_0^{{\psi'_i}} {\omega '\left( {\theta' } \right)d\theta' }  },&~\text{if}~ {\psi'_{i + 1}} < {\psi '_i}\\
\int_{{\psi '_i}}^{{\psi '_{i + 1}}} {\omega' \left( {\theta' } \right)d\theta' } .&~\text{otherwise} \hfill \\ 
\end{aligned}  \right.
\label{Workload}
\end{equation}
with \({\psi'_{N + 1}} = {\psi'_1}\) and \({m'_N} = {m'_0}\).
%\end{definition}
To balance the workload among subregions of (\(\Omega\)-\(O\))\(\cup\)\(\partial O\), the dynamics of partition bar is designed as
\begin{equation}  
{{\dot \psi '}_i} = {k_{\psi '}}\left( {m'_i} - {m'_{i-1}} \right)
\label{Partitioning dynamic}
\end{equation}
with the positive constant \({k_{\psi '}}\), \(i \in {I_N}\) and \({\psi'_i} = \,\bmod \,\left( {{\psi' _i},2\pi } \right)\). 
Similarly, the dynamics of the $i$-th agent is presented as
\begin{equation}  
 {{\dot p'}_i} = {u'_i},~\forall i\in {I_N}
\label{Control dynamic}
\end{equation}
with the control input \(u'_i\). 
\begin{remark}
Based on the above discussion, the density change of the subregion is much less than the density itself. According to Proposition~\ref{Proposition3.1}, we obtain \({\psi _i} = {\psi' _i}\). Since the sequential number of partition bars in (\(\Omega\)-\(O\))\(\cup\)\(\partial O\) and in \(\Xi\) are accurately corresponding, we have \({\dot \psi _i} = {\dot \psi' _i}\).
Since \({k_{\psi '}}\) and \({k_{\psi }}\) are positive constants, it follows from \({\dot \psi _i} = {k_\psi }\left( {{m_i} - {m_{i - 1}}} \right)\) and \({{\dot \psi '}_i} = {k_{\psi '}}\left( {m'_i} - {m'_{i-1}} \right)\) that 
\begin{equation}
\left( {{m_i} - {m_{i - 1}}} \right) \cdot \left( {{m'_i} - {m'_{i - 1}}} \right) > 0
\label{workload relationship}
\end{equation}
\end{remark}
From Remark~\ref{Estimation}, one obtains that the discretization error of the Delaunay triangulation can be estimated in the mapping process. Thus, we suppose that the non-smooth boundary of the obstacle can be approximated by the smooth curve. This allows us to analyze the geometric properties of the environment \(\Xi\) and (\(\Omega\)-\(O\))\(\cup\)\(\partial O\) in the view of Riemannian manifold. Furthermore, we design the control law according to the geometric properties of these regions.

\begin{definition}[Orientable Altas~\cite{Riemannian geometry}]
Suppose \(A\) is an Altas in \(\Xi\), we call \(A\) an orientable altas if any two coordinate charts in \(A\) are oriented-compatible.
\end{definition}

\begin{remark}
 \(\Xi\) and (\(\Omega\)-\(O\))\(\cup\)\(\partial O\) are topologically equivalent, and the orientation is a topological invariant. Suppose \(A\) is an orientable altas in \(\Xi\), there must exist an orientable altas in (\(\Omega\)-\(O\))\(\cup\)\(\partial O\).
\end{remark}

\begin{definition}[Orientable Manifold~\cite{Riemannian geometry}]
\(\Xi\) is an orientable manifold if there exists an oriented-compatible altas in \(\Xi\).
\label{definition3.2}
\end{definition}

\begin{theorem}
Suppose \(U\) is an open set in \({\mathbbm{R}^2}\), \(f:U \to {\mathbbm{R}^2}\), \(f \in {C^1}\left( U \right)\). \(y_0 \) is the regular value of \(f\). If \({y_0} \in f\left( U \right)\) and \(U \supset \Xi  = {f^{ - 1}}\left( {{y_0}} \right)\), then  \(\Xi \) is an orientable manifold.
\label{Theorem3.2}
\end{theorem}
\begin{proof}
The set \(\left( {{x^1},{x^2}} \right)\) is the coordinate representation in \(U_\alpha\). According to Regular Value Theorem ~\cite{Riemannian geometry}, \(U\) is a region in \({\mathbbm{R}^2}\), \(f \in {C^1}\left( U \right)\), if  \({y_0} \in f\left( U \right)\), \(U \supset \Xi  = {f^{ - 1}}\left( {{y_0}} \right)\), then \(\Xi \) must be a manifold. \(f' = \left( {{{\partial f}}/{{\partial {x^1}}},{{\partial f}}/{{\partial {x^2}}}} \right) = \nabla f\), \(rank\left( {f'} \right) \equiv 1\) in \(\Xi\) because \(y_0 \) is the regular value of \(f\), which means that \(\nabla f\ne 0\), \(\nabla f\) is a normal vector of \(\Xi\). \(\vec n = \nabla f/\left\| {\nabla f} \right\|\) forms a continuous unit normal vector field in \(\Xi\), and the continuity is guaranteed by \(f\in C^1(U)\). Suppose \(A\) is an atlas in \(\Xi\), \(\left( {{U_\alpha },{\phi _\alpha }} \right)\) and \(\left( {{U_\beta },{\phi _\beta }} \right)\) are coordinate charts in \(A\) and \({U_\alpha } \cap {U_\beta } \ne \emptyset \). \({\phi _\alpha }:{I^1} \to {U_\alpha }\), \({t^1} \mapsto \left( {{x^1}\left( {{t^1}} \right),{x^2}\left( {{t^1}} \right)} \right)\), \({\phi _\beta }:{I^1} \to {U_\beta }\), \({s^1} \mapsto \left( {{x^1}\left( {{s^1}} \right),{x^2}\left( {{s^1}} \right)} \right)\), where \({I^1} \) is the \(1\)-dimensional unit cube,  \(t^1\) is the coordinate representation in \({I^1} \). The tangent vector of the coordinate chart is represented as \(\partial /\partial {t^1} = \left( {\partial {x^1}/\partial {t^1},\partial {x^2}/\partial {t^1}} \right)\). We call \(\left( {{U_\alpha },{\phi _\alpha }} \right)\) a proper chart if \(\left( {\vec n,\partial /\partial {t^1}} \right)\) is oriented-compatible with \({\mathbbm{R}^2}\) standard basis, otherwise we call it an improper chart. For any improper chart, replace \( {t^1}\) with \(- ({t^1})\), then it becomes a proper chart. Suppose all coordinate charts in \(A\) are proper, then \(\det \left( {\vec n,\partial /\partial {t^1}} \right) > 0\), \(\det \left( {\vec n,\partial /\partial {s^1}} \right) > 0\). \(\partial {x^i}/\partial {s^1} = \left( {\partial {t^1}/\partial {s^1}} \right)\left( {\partial {x^i}/\partial {t^1}} \right)\), \(i = 1,2\), where \(\partial {t^1}/\partial {s^1} = {\phi' _{\alpha \beta }}\), \(\phi _{\alpha \beta }={\phi _\alpha^{-1} } \circ {\phi _\beta }\).  Then, we obtain
\[\left( {\vec n,\partial /\partial {t^1}} \right) = \left( {\vec n,\partial /\partial {s^1}} \right)\left( {\begin{array}{*{20}{c}}1&0\\0&{\partial {t^1}/\partial {s^1}} \end{array}} \right).\]
By taking the determinants on both sides of the equation, we have \(\det \left( {\partial {t^1}/\partial {s^1}} \right) > 0\), and it is equivalent to \(\det \left( {{\phi' _{\alpha \beta }}} \right) > 0\). It means that \(A\) is an oriented-compatible altas, according to Definition~\ref{definition3.2}, \(\Xi\) is an orientable manifold. 
\end{proof}
\(\Xi\) is embedded in $\mathbbm{R}^2$ and every point in \(\Xi\) is regular valued point, as illustrated by Theorem~\ref{Theorem3.2}. Since \(\Xi \) is an orientable manifold, we can define oriented-compatible coordinate charts within \(\Xi\). A group of \(N\) agents divides the coverage region \(\Xi\) into \(N\) partition regions, each with a corresponding partition bar. We select a coordinate chart for each partition region, ensuring that all selected charts are oriented-compatible. Without loss of generality, let \(\left( {{U_\alpha },{\phi _\alpha }} \right)\) and 
\(\left( {{U_\beta },{\phi _\beta }} \right)\) are two charts in adjacent partition region, with \({U_{\alpha \beta }}: = {U_\alpha } \cap {U_\beta } \ne \emptyset \) being the partition bar between two regions. Given the charts \(\left( {{U_\alpha },{\phi _\alpha }} \right)\), \(\left( {{U_\beta },{\phi _\beta }} \right)\) are oriented-compatible charts, there exists a transition function \({\phi_{\alpha\beta}}:\phi_\alpha^{- 1}\left({{U_{\alpha \beta }}}\right)\to\phi _\beta ^{ - 1}\left( {{U_{\alpha \beta }}} \right)\). For any \(t \in {U_{\alpha \beta }}\), it holds that \(\det {\phi '_{\alpha \beta }}\left( t \right) > 0\). \({\mathbbm{R}^2}  \supseteq \Xi\) is a 2-dimensional orientable manifold with boundary, then \(\partial \Xi \) is a 1-dimensional orientable manifold without boundary~\cite{Riemannian geometry}. This indicates that the boundaries of  \(\Xi\)  are also orientable. It is important to note that orientation is a topological invariant, which in turn implies that the boundaries of (\(\Omega\)-\(O\))\(\cup\)\(\partial O\) are orientable. 
\begin{remark}
Assume that \(\Xi\) is an orientable manifold with boundary, and let \(\{ \left( {{U_\alpha },{\phi _\alpha }} \right)\} \) is an orientable altas in \(\Xi\). For any point \(x\in\Xi\), there exits a diffeomorphism \(\tau :\Xi\to (\Omega-O)\cup\partial O\). Let \(x = {\tau ^{ - 1}}\left( y \right)\), there exists a chart \(({{U_\alpha },{\phi _\alpha }})\) in \(\Xi\) such that \(x \in {U_\alpha }\), \(y \in \tau \left( {{U_\alpha }} \right)\). For the chart \(\{ \left( {\tau \left( {{U_\alpha }} \right),{\phi _\alpha } \circ {\tau ^{ - 1}}} \right)\} \), these charts form an orientable altas within (\(\Omega\)-\(O\))\(\cup\)\(\partial O\), thereby making (\(\Omega\)-\(O\))\(\cup\)\(\partial O\) is an orientable manifold.
\end{remark}
Given the oriented-compatible property of the interior and boundary of the charts in \(\Xi\), the entire region \(\Xi\) is orientable, which suggests that we only need to analyze a single chart within \(\Xi\). We select \(\left( {{U_i },{\phi_i }} \right)\)  as a chart within the region \({E_i}\), which is enclosed by \({L_{in}},{L_{out}},{\psi _i},{\psi _{i+1}}\). Then we define the Riemannian metric \(\eta :{T_p}{E_i} \times {T_p}{E_i} \to R\) for any \(p \in {E_i}\). In the chart \(\left( {{U_i },{\phi_i }} \right)\), \(\eta\) has a matrix representation, with components denoted as \({\eta _{j k}}\). We use the notation \({\eta }\) to denote the matrix itself and \({\eta ^{j k}}\) for the components of the matrix inverse of the Riemannian metric \({\eta}\). Based on Equation \eqref{coordinate representative}, we represent the coordinate as \(\left( {\theta,\varphi } \right)\), with the original workspace coordinates represented as \(\left( {\theta ',\varphi '} \right)\). Through a straightforward calculation, we obtain
\begin{equation} 
\eta  = \left[\begin{array}{*{20}{c}}
  {{{\left( {R + r\cos \varphi } \right)}^2}}&0 \\ 
  0&{{r^2}{{\sin }^2}\varphi } 
\end{array}\right]
\label{Riemann Metric}
\end{equation}
Based on Equation \eqref{Jacobian of f}, we get \(\det \left( {{J_\tau }} \right) > 0\), which indicates the orientation of (\(\Omega\)-\(O\))\(\cup\)\(\partial O\) is consistent with that of \(\Xi\). Thus, we can employ a pull-back mapping to derive the Riemannian metric within (\(\Omega\)-\(O\))\(\cup\)\(\partial O\).
\begin{equation}
\eta ' = \left[\begin{array}{*{20}{c}}
{{g'_{11}}}&{{g'_{12}}}\\
{{g'_{21}}}&{{g'_{22}}}
\end{array}\right]
\end{equation}
with \({g'_{11}} = {\left( {R' + r'\cos \varphi '} \right)^2}{\left( {\frac{{\partial \theta '}}{{\partial \theta }}} \right)^2} + \left( {{{\left( {r'} \right)}^2}{{\sin }^2}\varphi '} \right){\left( {\frac{{\partial \varphi '}}{{\partial \theta }}} \right)^2}\), \({{g'_{12}} = g'_{21}} = {\left( {R' + r'\cos \varphi '} \right)^2}\left( {\frac{{\partial \theta '}}{{\partial \theta }}} \right)\left( {\frac{{\partial \theta '}}{{\partial \varphi }}} \right) + \left( {{{\left( {r'} \right)}^2}{{\sin }^2}\varphi '} \right)\left( {\frac{{\partial \varphi '}}{{\partial \theta }}} \right)\left( {\frac{{\partial \varphi '}}{{\partial \varphi }}} \right)\), \({g'_{22}} = {\left( {R' + r'\cos \varphi '} \right)^2}{\left( {\frac{{\partial \theta '}}{{\partial \varphi }}} \right)^2} + \left( {{{\left( {r'} \right)}^2}{{\sin }^2}\varphi '} \right){\left( {\frac{{\partial \varphi '}}{{\partial \varphi }}} \right)^2}\).

Consider a smooth manifold \(\Xi\) with boundary, which is equipped with its standard Riemannian metric \(\left( {{\mathbbm{R}^2},\eta} \right)\). In this work, the coverage region \(\Xi\) is non-convex, yet the partition region \({E_i} \subset \Xi \) is path connected. Inspired by the work in~\cite{Unknown13}, we introduce a length metric structure \({d_l}\left( { p, q } \right)\) on \({E_i}\) to navigate an agent to an optimal centroid while avoiding collisions with the boundary of the manifold and obstacles, and \({d_l}\left( { p, q } \right)\) is equal to the infimum of the lengths of all rectifiable paths connecting \(p\) and \(q\). This ensures that \({E_i}\) is a path metric space, and consequently, a complete metric space. Since \({E_i} \) is a subset of the orientable manifold \(\Xi\), \(\Xi\) itself is a path metric space equipped with the length structure \({d_l}\left( { p, q } \right)\). We typically utilize the Riemannian metric to induce the length metric \({d_l}\left( { p, q } \right)\), distinguishing among the definitions in the same chart or different charts. Given an agent position \(p\) and a target position \(q\), with \(Q = {\phi _\alpha }\left( q \right)\), \(P = {\phi _\alpha }\left( p \right)\), and a path \(\gamma:\left[ {p,q} \right] \to \Xi \), where \(\gamma \left( t \right) = \left( {{x^1}\left( t \right), \cdots,{x^n}\left( t \right)} \right)\). {blue}We define the path length \(L\left( \gamma \right)\) as the length of \(\gamma\). Let \(L\left( \gamma \right)\) denote the integral of the square root of the sum of the products of the coefficients \({\eta _{ij}}( {\gamma \left( t \right))}\) and the squares of the derivatives of the coordinates \(x^i\) with respect to \(t\), over the interval \(\left[ {p,q} \right]\), that is,
$$
L\left( \gamma \right) = \int_p^q {\sqrt {\sum\nolimits_{i,j = 1}^n {{\eta _{ij}}\left( {\gamma \left( t \right)} \right)({{d{x^i}}}/{{dt}})({{d{x^j}}}/{{dt}})} } } dt.
$$
Furthermore, the length metric \({d_l}( { \cdot, \cdot })\) is defined as the infimum of \(L\left( \gamma \right)\), given by
 \begin{equation} 
{d_l}\left( { p, q } \right)= \inf \{ L\left( \gamma \right)\}. 
\label{Length metric}
\end{equation}
The metric satisfies the fundamental metric properties. 1) Non-negativity: the length of any curve is non-negative. 2)
Symmetry: the distance between points \(p\) and \(q\) is equal to the distance between \(q\) and \(p\). 3) Triangular inequality:  
For any three points \(x,y,z\) in the curve connecting \(p\) and \(q\), the infimum of the distance between \(x\) and \(z\) is less than or equal to the sum of the infimums of the distances between \(x\) to \(y\) and \(y\) to \(z\), which equals to \(\inf d(x,z) \leq \inf d(x,y) + \inf d(y,z)\). It implies that those points satisfy \({d_l}(x,z) \le {d_l}(x,y) + {d_l}(y,z)\). Thus, \({d_l}\left( { \cdot, \cdot } \right)\) constitutes a metric. Similarly, using \({\eta '}\) to induce the length metric in the original workspace (\(\Omega\)-\(O\))\(\cup\)\(\partial O\), we obtain
\begin{equation} 
{d_l}\left( { p', q' } \right)= \inf \left\{ \int_{p'}^{q'} {\sqrt {\sum\nolimits_{i,j = 1}^n {{{\eta '}_{ij}}\left( {\gamma '\left( t \right)} \right)\frac{{d{{x'}^i}}}{{dt}}\frac{{d{{x'}^j}}}{{dt}}} } } dt\right\},
\end{equation}
where \(\gamma ':\left[ {p,q} \right] \to \Omega {\rm{ }}\) and \(\gamma '\left( t \right) = \left( {{{x'}^1}\left( t \right), \cdots ,{{x'}^n}\left( t \right)} \right)\).
Thus, we can define a length metric \({d_l}\left( { \cdot, \cdot } \right)\) on \(\Xi\), and for a chart \(\left( {{U_\alpha },{\phi _\alpha }} \right)\) in \({E_i}\). We have \(d_l^C\left( {p,q} \right) = {d_l}\left( {\phi _\alpha ^{ - 1}\left( P \right),\phi _\alpha ^{ - 1}\left( Q \right)} \right)\) when \(p,q \in {U_\alpha }\) and \(d_l^{\tilde C}\left( {p,q} \right) = {d_l}\left( {p,\phi _\alpha ^{ - 1}\left( Q \right)} \right)\) when \( p\in \Omega \backslash {U_\alpha }\),  \(q \in {U_\alpha }\). In a compact region \(\Xi\), the existence of Cauchy sequences implies the existence of the negative gradient at \(q\), which in turn suggests the existence of a minimal path from \(q\). To analyze the existence of minimum path over the entire region \(\Xi\), it is necessary to analyze the geometry properties of subregions (i.e. convexity, connectedness) after defining the metric structure.   
\begin{theorem}
The partition region \({E_i}\) is a geodesically convex region.
\label{geodesicallyConvex}
\end{theorem}
\begin{proof}
\(\left( {E_i,{d_l}} \right)\) is a smooth manifold with a boundary, which indicates the existence of the partition bar \({\psi_i}\) in the region. Moreover, it is a complete metric space with a length metric.  By the Hopf-Rinow Theorem~\cite{Riemannian geometry}, there exists a minimal path (though not necessarily unique) connecting any pair of points. Consider a point \(p \in {E_i}\) and another point \(q \in {{E_i}-\partial{E_i}}\), to ensure the uniqueness of the minimal path, the cut locus \({C_p}\) is supposed to be a Lebesgue-measure-zero set.  For any point \(p \in {E_i}\), since \({E_i}\) can be considered as a subset of a smooth, complete manifold of the same dimension, according to the Theorem 6.2 in~\cite{Cut loci85}, the cut locus \({C_p}\) is a Lebesgue zero measure set in \({E_i}\). Consequently, there exists only one minimal path connecting \(p\) and \(q\), which is entirely within this partition region, indicating that the partition region \(E_i\) is a geodesically convex region.  
\end{proof}

The partition region \(E_i\) is a geodesically convex region, implying the existence of a minimal path in \(E_i\). In reality, the agent will pass through different sub-regions. Thus, it is imperative to ensure the geodesically convex property of the entire region \(\Xi\).
\begin{corollary}
The mapped region \(\Xi\) is a geodesically convex region.
\label{Corollary3.3.1}
\end{corollary}
\begin{proof}
Since \({E_i}\) is a geodesically convex region and \({E_i} \subset \Xi \subseteq \sum\nolimits_{i = 1}^N {{E_i}} \). Note that the partition bar \({\psi _i} = {E_{i - 1}} \cap {E_i}\) is a continuous and smooth set, which means that the Riemannian metric carried by each region is continuous at the splicing point, and the geodesic is coincident when the partition bar \({E_i}\) and \({E_{i-1}}\) are spliced together. Thus, the minimal path passing through different partition regions is entirely within the spliced region \(\Xi\). This demonstrates that the mapped region \(\Xi\) is also a geodesically convex region.   
\end{proof}

\begin{remark}
The diffeomorphism \(\tau\) satisfies \(\tau :{\mathop{\rm int}} \left( {\left( {\Omega  - O} \right) \cup \partial O} \right) \to {\mathop{\rm int}}  \Xi  \) and \(\tau :\partial \left( {\left( {\Omega  - O} \right) \cup \partial O} \right) \to \partial \Xi \), where \({\mathop{\rm int}} \Xi \) denotes the inner region of \(\Xi\), and 
\(\partial \Xi \) represents the outer region of \(\Xi\). Thus, the preimage of minimal path in \(\Xi\) is entirely located in (\(\Omega\)-\(O\))\(\cup\)\(\partial O\).
\end{remark}

The partition region \(\Xi\) is a geodesically convex  region, which is complete as a
metric space with the length metric \({d_l}\left( { \cdot, \cdot } \right)\) induced by the Riemannian
metric $\eta$. The above illustrations and definitions open the door to defining the relationship between the tangent of a minimizing path and the derivative of the distance function 
${d_l}({\cdot,\cdot})$. According to Theorem~\ref{geodesicallyConvex}, the cut locus \({C_p}\) is a Lebesgue zero measure set in \(\Xi\). For an arbitrarily chosen  chart \(\left( {{U_\alpha },{\phi_\alpha }} \right)\), since 
\[\sqrt{{\eta_{ij}}({{\phi_\alpha }(q)})({d{x^i}/dt})( {d{x^j}/dt})}=\frac{{{\eta_{ij}}({{\phi_\alpha}(q)})( {d{x^i}/dt})({d{x^j}/dt})}}{{\sqrt{{\eta_{mn}}({{\phi_\alpha }(q)})({d{x^m}/dt})({d{x^n}/dt})}}},\]
it holds that
 \begin{equation} 
\frac{\partial }{{\partial {q^i}}}d_{_l}^C\left( {{\phi _\alpha }\left( p \right),{\phi _\alpha }\left( q \right)} \right)=\frac{{{\eta _{ij}}\left( {{\phi _\alpha }\left( q \right)} \right)z_{p{\phi _\alpha }\left( q \right)}^j}}{{\sqrt {{\eta _{mn}}\left( {{\phi _\alpha }\left( q \right)} \right)z_{p{\phi _\alpha }\left( q \right)}^mz_{p{\phi _\alpha }\left( q \right)}^n} }},
\label{equation19}
\end{equation}
where \({z_{p{\phi _\alpha }\left( q \right)}} = {\left[ {z_{p{\phi _\alpha }\left( q \right)}^1,z_{p{\phi _\alpha }\left( q \right)}^2, \cdots,z_{p{\phi _\alpha }\left( q \right)}^N} \right]^T}\) is the coefficient vector in coordinate chart \(\left( {{U_\alpha },{\phi _\alpha }} \right)\) of the tangent vector at \(q\) to the shortest path connecting \(p\)  to \(q\). Inspired by the work~\cite{SBhattacharya14}, suppose the position of agent \(p\) is not within the chart \(\left( {{U_\alpha },{\phi _\alpha }} \right)\), we obtain
\begin{equation} 
\frac{\partial }{{\partial {q^i}}}d_{_l}^{\tilde C}  \left( {p,{\phi _\alpha }\left( q \right)} \right) =\frac{{{\eta _{ij}}\left( {{\phi _\alpha }\left( q \right)} \right)z_{p{\phi _\alpha }\left( q \right)}^j}}{{\sqrt {{\eta _{mn}}\left( {{\phi _\alpha }\left( q \right)} \right)z_{p{\phi _\alpha }\left( q \right)}^mz_{p{\phi _\alpha }\left( q \right)}^n} }},
\label{equation20}
\end{equation}
where \({z_{p{\phi _\alpha }\left( q \right)}} = {\left[ {z_{p{\phi _\alpha }\left( q \right)}^1,z_{p{\phi _\alpha }\left( q \right)}^2, \cdots,z_{p{\phi _\alpha }\left( q \right)}^N} \right]^T}\) is the coefficient vector of the tangent vector at \(q\) to the shortest path  connecting \(p \in \Xi \backslash {U_\alpha }\) to \(q \in {U_\alpha }\). The following remark gives a brief explanation of  Equation~\eqref{equation20}. 
\begin{remark}
Let \({\bf{q}}=\phi _\alpha(q)\). Consider a minimal path \({\gamma _{pq}} = {\gamma _{pw}} \cup {\gamma _{wq}}\), where \(w \in {B_q}\), \({B_q}\) is a open ball of \(q\). For any \({\bf{q}} \in {\phi _\alpha(B_q) }\), \({d_l^{\tilde D}}\left( {p,\bf{q}} \right) \le {d_l}\left( {p,w} \right) + {d_l^{\tilde D}}\left( {w,\bf{q}} \right)\), the equality holds when \(p\),\(w\) and \({\phi _\alpha^{-1}(\bf{q})}\) lie on the same short path. According to Corollary~\ref{Corollary3.3.1}, the mapped region \(\Xi\) is a geodesically convex region, which implies that the cut locus is Lebesgue zero measure in chart  \(\left( {{U_\alpha },{\phi _\alpha }} \right)\).  This implies that the differentials of the distance functions at \(\bf{q}\) should be the same, 
which leads to \[\frac{\partial }{{\partial \bf{q}}}{d^{\tilde D}}\left( {p,\bf{q}} \right) = \frac{\partial }{{\partial \bf{q}}}{d^{\tilde D}}\left( {w,\bf{q}} \right),\]
where \({\gamma _{pw}}\) is a minimal path, and \({\gamma _{wq}}\) is a minimal geodesic. According to~\cite{Pimenta08}, the coefficient vector of the tangent vector at \(\bf{q}\) is the same, which means that \({z_{p\bf{q}}} = {z_{w\bf{q}}}\). Consequently, the generalization of Equation~\eqref{equation20} has the same form as Equation~\eqref{equation19}.
\end{remark}
These definitions essentially reveal the relationship between the gradient of the length metric function and the tangent of the minimal path connecting two points. Particularly, we explore the situation where the agent position is not within the designated chart. This scenario serves as a suitable backdrop for examining the challenge of steering the agent to the optimal centroid of the sub-region. The gradient of the length metric function in the region (\(\Omega\)-\(O\))\(\cup\)\(\partial O\) is defined identically to Equations~\eqref{equation19} and \eqref{equation20}, because (\(\Omega\)-\(O\))\(\cup\)\(\partial O\) and \(\Xi\) are topologically equivalent. The performance indexes for multi-agent coverage problem in (\(\Omega\)-\(O\))\(\cup\)\(\partial O\) is designed as:
\begin{equation} 
J'\left( {\psi ',{\bf{p'}}} \right) = \sum\limits_{i = 1}^N {\int_{{{E'}_i}\left( {\psi '} \right)} {{f_i}\left( {{d_l}\left( {q',{{p'}_i}} \right)} \right)\rho \left( {q'} \right)dq'} } 
\label{equation22}
\end{equation}
with \({\bf{p}} = {[{p_1},{p_2}, \cdots ,{p_N}]^T}\) and \(\psi  = {[{\psi _1},{\psi _2}, \cdots ,{\psi _N}]^T}\), \({\bf{p'}} = {[{p'_1},{p'_2}, \cdots ,{p'_N}]^T}\). \(J\left( {\psi ,{\bf{p}}} \right)\) and \(J'\left( {\psi ',{\bf{p'}}} \right)\)are the performance index in \(\Xi\) and (\(\Omega\)-\(O\))\(\cup\)\(\partial O\), respectively. 
In this work, for \(\Xi\) or (\(\Omega\)-\(O\))\(\cup\)\(\partial O\), we work with the following form 
\begin{equation} 
f_i(x) = {x^2}. 
\end{equation}
The rationality of the choice is illustrated in~\cite{SBhattacharya14}. Essentially, our core objective is to design a distributed control strategy to address the optimization problem in (\(\Omega\)-\(O\))\(\cup\)\(\partial O\), while ensuring that the optimization criteria in \(\Xi\) are also met. Hence, we present the optimization problem as follows: 
\begin{equation} 
\mathop {\min }\limits_{\psi ',{\bf{p'}}} J'\left( {\psi ',{\bf{p'}}} \right)
\label{minJ'}
\end{equation}
which should be solved in (\(\Omega\)-\(O\))\(\cup\)\(\partial O\). Subsequently, we focus on the stability and convergence properties of multi-agent partition dynamics.

\begin{theorem}
Partition dynamics \eqref{Partitioning dynamic} equalizes workload on each subregion of the original workspace.
\end{theorem}

\begin{proof}
Construct the Lyapunov function candidate as follows 
\begin{equation}
V\left( \psi'  \right) = \frac{1}{2}\sum\limits_{i = 1}^N {{{\left( {{m'_i} - \bar m'} \right)}^2}}, 
\label{LyapunovFunction}
\end{equation}
where \(\bar m'=\frac{1}{N}\int_0^{2\pi } {\omega \left( {\theta '} \right)} d\theta '\) and \({V(\psi^{'}) \ge 0}\). From Equation\eqref{Workload}, we understand that \({\dot m'_i} = {\dot \psi '_{i + 1}}\omega \left( {{\psi '_{i + 1}}} \right) - {\dot \psi '_i}\omega \left( {{\psi '_i}} \right)\). Given that \({\psi '_{N + 1}} = {\psi '_1}\), we obtain \(\sum\nolimits_{i = 1}^N {{{\dot m'}_i}}  = 0\). According to Proposition~\ref{Proposition3.1} and the diffeomorphism \(\tau\) is conformal, we obtain \({\psi _i} = {\psi' _i}\). 
Note that the manifold (\(\Omega\)-\(O\))\(\cup\)\(\partial O\) is orientable, which implies that these charts are oriented-compatible, and we can choose an orientation of (\(\Omega\)-\(O\))\(\cup\)\(\partial O\). According to Theorem~\ref{Theorem 3.1}, the relative position of partition bars in (\(\Omega\)-\(O\))\(\cup\)\(\partial O\) and in \(\Xi\) are exactly corresponding, and we deduce that \({\dot \psi '_{N + 1}} =  {{\dot \psi }_{N + 1}}\). Adhering to the cyclic ordering rule, we reinterpret the last term as the first, ensuring that it satisfies \(\sum\nolimits_{i = 1}^N {{m'_i}{{\dot \psi }_{i + 1}}\omega \left( {{\psi _{i + 1}}} \right)}  = \sum\nolimits_{i = 1}^N {{m'_{i - 1}}{{\dot \psi }_i}\omega \left( {{\psi _i}} \right){\rm{ }}} \). Since the collision avoidance of split bars in \(\Xi\) is guaranteed with partition dynamics~\eqref{partitioning dynamic}~\cite{zhai23} and the conformal map \(\tau\) is a diffeomorphism, the collision avoidance of split bars in (\(\Omega\)-\(O\))\(\cup\)\(\partial O\) is guaranteed. Moreover, the time derivative of the Lyapunov function along the system trajectory is given by
\begin{align*}
{dV\left( \psi' \right)} \mathord{\left/{\vphantom {{dV\left( \psi  \right)} {dt}}} \right.\kern-\nulldelimiterspace} {dt} 
&=\sum\limits_{i = 1}^N {\left( {{m'_i} - \bar m'} \right){{\dot m'}_i}} \\
 &= \sum\limits_{i = 1}^N {{m'_i}{{\dot m'}_i} - \bar m'\sum\limits_{i = 1}^N {{{\dot m'}_i}} } \\
 &= \sum\limits_{i = 1}^N {{m'_i}{{\dot m'}_i}} \\
 &= \sum\limits_{i = 1}^N {{m'_i}\left( {{{\dot \psi '}_{i + 1}}\omega \left( {{\psi '_{i + 1}}} \right) - {\dot \psi '_i}\omega \left( {{\psi '_i}} \right)} \right)} \\
 &= \sum\limits_{i = 1}^N {{m'_i}\left( {{{\dot \psi }_{i + 1}}\omega \left( {{\psi' _{i + 1}}} \right) - {{\dot \psi }_i}\omega \left( {{\psi' _i}} \right)} \right)} \\
 &=  - \sum\limits_{i = 1}^N {\left( {{m'_i} - {m'_{i - 1}}} \right){{\dot \psi }_i}\omega \left( {{\psi' _i}} \right)} \\
 &=  - {K_\psi }\sum\limits_{i = 1}^N {\left( {{m'_i} - {m'_{i - 1}}} \right)\left( {{m_i} - {m_{i - 1}}} \right)\omega \left( {{\psi' _i}} \right)} 
\end{align*}
In light of \eqref{workload relationship}, we obtain \({{\left( {{{m'}_i} - {{m'}_{i - 1}}} \right)\left( {{m_i} - {m_{i - 1}}} \right)}>0}\). Since \({{K_\psi }>0}\), \({{\omega \left( {{\psi' _i}} \right)}>0}\), it implies that \({{dV\left( \psi  \right)} \mathord{\left/{\vphantom {{dV\left( \psi  \right)} {dt}}} \right.\kern-\nulldelimiterspace} {dt}}<0\). Therefore, the system is asymptotically stable, the state of the system will approach the equilibrium point as time goes to infinity, which implies that \({m'_1} = {m'_2} =  \ldots  = {m'_N}\).
\end{proof}
\begin{remark}
Analyzing the convergence rate of the system helps to find the supremum of \(V\left( {\psi '} \right)\). In \(\Xi\), the Lyapunov function is \(V\left( \psi  \right) = \frac{1}{2} {\sum\nolimits_{i = 1}^N {{{\left( {{m_i} - \bar m} \right)}^2}} }\), equitable workload partitions can be achieved in \(\Xi\) with an exponential convergence rate~\cite{zhai23}. The convergence rate is also consistent in (\(\Omega\)-\(O\))\(\cup\)\(\partial O\) because the two Lyapunov functions are of the same form.
\end{remark}

After analyzing the system stability, our attention is turned to the convergence of multi-agent partition dynamics. First, we analyze the existence of the optimization index. Secondly, we prove the existence of optimal solutions to 
Problem~\({\nabla _{{p_i}}}J'=0\). Finally, we prove that there always exist optimal solutions to Problem~\ref{minJ'} . 

\begin{theorem}
There always exist optimal and reachable solutions to Problem~\eqref{minJ'}.
\label{Theorem27}
\end{theorem}

\begin{proof}
Since the cut locus in (\(\Omega\)-\(O\))\(\cup\)\(\partial O\) is Lebesgue measure zero, \(J'\) is arbitrarily close to be a \({C^1}\) function~\cite{Cut loci85}. The existence of \({\nabla _{{p_i}}}J'\) is therefore guaranteed.  It is obviously that \({\nabla _{{p_i}}}J'\) resides in \(T_{{P_i}}^*((\Omega\)-\(O\))\(\cup\)\(\partial O )\) and \({\nabla _{{p_i}}}J' \cdot {\eta ^{il}}\left( {{p_i}} \right)\) resides in \(T_{{P_i}}((\Omega\)-\(O\))\(\cup\)\(\partial O )\). \({\nabla _{{p_i}}}J'\) can be viewed as a vector field in (\(\Omega\)-\(O\))\(\cup\)\(\partial O\), which is dependent on the Riemannian metric \({\eta '}\). Note that there exists an atlas such that \(((\Omega  - O) \cup \partial O) \subset \{ \left( {{V_i},{\xi _i}} \right)\} \), where every chart in this atlas is local Euclidean. Thus, one can find a \({C^1}\) diffeomorphism \({\xi _i}\) that can map the chart \({V_i}\) to \({I^2} = [0,1] \times [0,1]\). We denote \({{\tilde p}_i} = {\xi _i}\left( {{p_i}} \right)\) and \({\nabla _{{{\tilde p}_i}}}H = {\xi _i}\left( {{\nabla _{{p_i}}}J'} \right)\). Moreover, \({\xi _i}\) satisfies the \({\xi _i}\)-relative condition, which means that \({\xi _i}\) preserves the zero point of \({\nabla _{{p_i}}}J'\) and \({\nabla _{{{\tilde p}_i}}}H\)~\cite{Riemannian geometry}. There is no  local minimum of \({\nabla _{{p_i}}}J'\) on \(\partial \Omega \) or \(\partial O \) if the length metric \({d_l}\left( { \cdot, \cdot } \right)\) is used to compute \({\nabla _{{p_i}}}J'\) as per the work in~\cite{SBhattacharya14}. Due to the bijective property of \({\xi _i}\), there does not exist a local minimum of \({\nabla _{{{\tilde p}_i}}}H\) on \(\partial I \) if the length metric \({d_l}\left( { \cdot, \cdot } \right)\) is used to compute \({\nabla _{{{\tilde p}_i}}}H\). According to the generalization of Bolzano’s Theorem~\cite{bolzanos02}, solutions to \({\nabla _{{{\tilde p}_i}}}H=0\) exist on \({I^2}\) if \(\left\langle {{\nabla _{{{\tilde p}_i}}}H,{{\tilde p}_i}} \right\rangle  > 0\) for any \({{\tilde p}_i} \in \partial {I^2}\). Given that \({\xi _i}\) satisfies the \({\xi _i}\)-relative condition, solutions to \({\nabla _{{p_i}}}J'=0\) exist.
 Let \({f_i}\left( x \right) = {x^2}\) in \(\Xi\). According to~\eqref{equation19}, since \(p_i\) and \({C_p}_i\) are Lebesgue-measure-zero, we obtain
\begin{align}
\begin{split}
{\nabla _{{p_i}}}J &= 2\int_{{E_i}\left( \psi  \right)} {{d_l}\left( {q,{p_i}} \right)\frac{\partial }{{\partial {p_i}}}} d_l^{ C}(q,{p_i})\rho \left( q \right)dq \\ 
&= 2\int_{{E_i}\left( \psi  \right) - \left( {{p_i} \cup {C_{{p_i}}}} \right)} {{d_l}\left( {q,{p_i}} \right)\frac{{{\eta _{ij}}\left( {{p_i}} \right)z_{q{p_i}}^j}}{{\sqrt {{\eta _{mn}}\left( {{p_i}} \right)z_{q{p_i}}^mz_{q{p_i}}^n} }}} \rho \left( q \right)dq,
\label{deltaJ}
\end{split}
\end{align}
where \(q\) is the target position, and \({z_{qp}} = {\left[ {z_{qp}^1,z_{qp}^2, \cdots ,z_{qp}^N} \right]^T}\) is the coefficient vector of the tangent vector at \(p_i\) to the shortest geodesic connecting \(p_i\) and \(q\). 
Through the application of the diffeomorphism \(\tau \), the point \(\ {p_i^*} \) is identified as the optimal centroid within \(\Xi\). Due to \({\tau ^{ - 1}}\left( {p_i^*} \right) = {\left( {p_i^*} \right)^\prime }\), we derive the gradient of performance index in (\(\Omega\)-\(O\))\(\cup\)\(\partial O\) as follows
\begin{equation} 
{\nabla _{{p'_i}}}J' = 2\int_{{E'_i}\left( \psi  \right) - \left( {{p'_i} \cup {C_{{p'_i}}}} \right)} {{d_l}\left( {{{\left( {p_i^*} \right)}^\prime },{p'_i}} \right)\frac{{{{\eta '}_{ij}}\left( {{p'_i}} \right)z_{{{\left( {p_i^*} \right)}^\prime }{p'_i}}^j}}{{\sqrt {{{\eta '}_{mn}}\left( {{p'_i}} \right)z_{{{\left( {p_i^*} \right)}^\prime }{p'_i}}^mz_{{{\left( {p_i^*} \right)}^\prime }{p'_i}}^n} }}} \rho \left( {{{\left( {p_i^*} \right)}^\prime }} \right)d{\left( {p_i^*} \right)^\prime },
\label{PerformanceIndex}
\end{equation}
where \({z_{\left( {p_i^ * } \right)'{p'_i}}} = {\left[ {z_{\left( {p_i^ * } \right)'{p'_i}}^1,z_{\left( {p_i^ * } \right)'{p'_i}}^2, \cdots ,z_{\left( {p_i^ * } \right)'{p'_i}}^N} \right]^T}\) is the coefficient vector of the tangent vector at \(p'_i\) to the shortest geodesic connecting \(p'_i\) and    \(\left( {p_i^ * } \right)'\). Equations \eqref{deltaJ} and \eqref{PerformanceIndex} reveal the relationship between the derivate of the distance function and the performance index.
According to~\cite{zhai23}, there exists a differentiable function \(\zeta :[0,2\pi ] \to [0,2\pi ]\) such that \[\int_{\psi '}^{\zeta \left( {\psi '} \right)} {\omega \left( {\theta '} \right)} d\theta ' = \bar m', \forall\psi'\in [0,2\pi ]\]
with \(\omega \left( {\theta '} \right) = 
\int_{{\varphi '_{in}}}^{{\varphi'_{out}}} {{\tau ^{ - 1}}\left( {\rho (\theta ,\varphi )} \right)} {\left( {R' + r'\cos \varphi '} \right)^2}\sin \varphi 'd\varphi '\). This function satisfies \({\psi _i} = \zeta \left( {{\psi _{i - 1}}} \right)\) when equitable workload partition is completed. Given that (\(\Omega\)-\(O\))\(\cup\)\(\partial O\) is orientable, these charts are oriented-compatible. Consequently, by examining a single chart and understanding its properties, we can readily extend our insights to the entire region under consideration. According to the analysis, we aim to demonstrate the practicality of consolidating the analysis of the performance index into a single chart (For simplicity, we select the first chart), thereby reducing the complexity of all charts to a single chart. Based on \eqref{equation22}, the performance index can be rewritten as follows:
\begin{align*}
{J'\left( {\psi ',{\bf{p'}}} \right)}&= {\sum\limits_{i = 1}^N {\int_{{E'_i}\left( {\psi '} \right)} {{f_i}\left( {{d_l}\left( {q',{{p'}_i}} \right)} \right)\rho \left( {q'} \right)dq'} }}\\ & = {\sum\limits_{i = 1}^N {\int_{{\psi '_i}}^{\zeta \left( {{\psi '_i}} \right)} {\int_{{\varphi '_{in}}}^{{\varphi '_{out}}} {{f_i}\left( {{d_l}\left( {q',{p'_i}} \right)} \right)\rho \left( {\theta ',\varphi '} \right){{\left( {R' + r'\cos \varphi '} \right)}^2}\sin \varphi 'd\varphi 'd\theta '} } }}\\  &= {\sum\limits_{i = 1}^N {\int_{{\zeta ^{i - 1}}\left( {{\psi '_1}} \right)}^{{\zeta ^i}\left( {{\psi '_1}} \right)} {\int_{{\varphi '_{in}}}^{{\varphi '_{out}}} {{f_i}\left( {{d_l}\left( {q',{p'_i}} \right)} \right)\rho \left( {\theta ',\varphi '} \right){{\left( {R' + r'\cos \varphi '} \right)}^2}\sin \varphi 'd\varphi 'd\theta '} } }}, 
\end{align*}
where \({\zeta ^i}\) represents the composition of \(\zeta\) with itself \(i\) times, \(\forall i \in {I_N}\), and \({\psi' _1} \in [{\left( {{\psi' _1}} \right)^*},{\left( {{\psi' _2}} \right)^*}]\), with \({\left( {{\psi '_1}} \right)^*}\) denoting the first partition bar in  (\(\Omega\)-\(O\))\(\cup\)\(\partial O\). The region \({E'_1} \cup {\left( {{\psi' _1}} \right)^*} \cup {\left( {{\psi' _2}} \right)^*}\) is compact, and \({J'\left( {\psi'_1,{\bf{p'}}} \right)}\) is a continuous function. By the extreme value theorem, a continuous function on a compact set must attain an extreme value. Since \(\tau\) maps a compact set to a compact set, the existence of optimal solutions in \(\Xi\) is guaranteed, thereby ensuring the existence of optimal solutions to Problems~\eqref{minJ'}. 
\end{proof}
Considering the information transfer time induced by the mapping process. According to Theorem~\ref{Theorem27}, there always exist optimal solutions to Problems~\eqref{minJ'}. Then, we need to illustrate the relationship of optimal performance between \(\Xi\) and (\(\Omega\)-\(O\))\(\cup\)\(\partial O\).
 \begin{theorem}
The performance index \(J'\) attains the minimum in the original space (\(\Omega\)-\(O\))\(\cup\)\(\partial O\) when \(J\) reaches the minimum in the mapped space \(\Xi\).
\label{TheoremJ}
\end{theorem}
\begin{proof}
Consider the performance index \[{J'\left( {\psi ',{\bf{p'}}} \right)} = {\sum\limits_{i = 1}^N {\int_{{E'_i}\left( {\psi '} \right)} {{f_i}\left( {{d_l}\left( {q',{p'_i}} \right)} \right)\rho \left( {q'} \right)dq'} }}.\] 
Particularly, when \( {f_i}\left( {{d_l}\left( {q',{{p'}_i}} \right)} \right) \ge 0\), and \(\rho \left( {q'} \right) \ge 0\), the variation trend \(J'\) and \(J\) are identical. \(\tau\) acts upon the partition bar \(\psi'\) and the agent position \(\bf{p'}\), transforming them into \(\psi\) and  \(\bf{p}\). In the mapped space \(\Xi\), \(J\) is continuous with respect to the partition bar \(\psi\) and the agent position \(\bf{p}\). Since \(\tau\) is a diffeomorphism, it constitutes a continuous mapping, \(\tau\) and \(\tau^{-1}\) are differentiable. Furthermore, the Jacobian of \(\tau\) and \(\tau^{-1}\) is invertible. Consequently, \(J'\) is continuous with respect to the partition bar \(\psi'\) and the agent position \(\bf{p}'\) in the compact region \((\Omega  - O) \cup \partial O\). Suppose \({J'\left( {\psi_\alpha ',{\bf{p}_\alpha'}} \right)}\) and \({J\left( {\psi_\alpha ,{\bf{p}_\alpha}} \right)}\) are the minimum of \(J'\) and \(J\), respectively. Due to the point-to-point correspondence of diffeomorphism \(\tau\), we obtain \(\tau (\psi_\alpha ') = \psi_\alpha \), \(\tau (\bf{p}_\alpha') = \bf{p}_\alpha\). Subsequently, we obtain \(J'\left( {{\psi _{\alpha '}},{{\bf{p}}_{{\alpha '}}}} \right) = (J \circ \tau )\left( {{\psi _{\alpha '}},{{\bf{p}}_{{\alpha '}}}} \right)\). Then, we take the derivative of both sides of the equation simultaneously, which leads to
\(\dot J'\left( {{\psi _{\alpha '}},{{\bf{p}}_{{\alpha '}}}} \right) = \dot J\left( {\tau \left( {{\psi _{\alpha '}},{{\bf{p}}_{{\alpha '}}}} \right)} \right) \cdot \dot \tau \left( {{\psi _{\alpha '}},{{\bf{p}}_{{\alpha '}}}} \right)\). 
Since \(\dot J'\left( {{\psi _{\alpha '}},{{\bf{p}}_{{\alpha '}}}} \right) = 0\) and \(\dot \tau \left( {{\psi _{\alpha '}},{{\bf{p}}_{{\alpha '}}}} \right)\) is invertible, we have \(\dot J\left( {\tau \left( {{\psi _{\alpha '}},{{\bf{p}}_{{\alpha '}}}} \right)} \right) = 0\). Thus, \({\tau \left( {{\psi _{\alpha '}},{{\bf{p}}_{{\alpha '}}}} \right)}\) is a minimum point of \(J\). On the other hand, let \(\left( {{\psi _\alpha },{{\bf{p}}_{\bf{\alpha }}}} \right)\) is a minimum point of \(J\), which satisfies \(\left( {{\psi _{\alpha '}},{{\bf{p}}_{{\alpha '}}}} \right) = {\tau ^{ - 1}}\left( {{\psi _\alpha },{{\bf{p}}_{\alpha }}} \right)\). For the equation \(J\left( {{\psi _\alpha },{{\bf{p}}_{\bf{\alpha }}}} \right) = (J' \circ {\tau ^{ - 1}})\left( {{\psi _\alpha },{{\bf{p}}_{\bf{\alpha }}}} \right)\), we take the derivative in the point \({{\bf{p}}_{\alpha }}\) and obtain 
\(\dot J\left( {{\psi _\alpha },{{\bf{p}}_{\bf{\alpha }}}} \right) = \dot J'\left( {{\tau ^{ - 1}}\left( {{\psi _\alpha },{{\bf{p}}_{\bf{\alpha }}}} \right)} \right) \cdot {{\dot \tau }^{ - 1}}\left( {{\psi _\alpha },{{\bf{p}}_{\bf{\alpha }}}} \right).\)
Since \({{\dot \tau }^{ - 1}}\left( {{\psi _\alpha },{{\bf{p}}_{\bf{\alpha }}}} \right)\) is invertible with \(\dot J\left( {{\psi _\alpha },{{\bf{p}}_{\bf{\alpha }}}} \right) = 0\), we have \(\dot J'\left( {{\tau ^{ - 1}}\left( {{\psi _\alpha },{{\bf{p}}_{\bf{\alpha }}}} \right)} \right) = \dot J'\left( {{\psi _{\alpha '}},{{\bf{p}}_{{\bf{\alpha '}}}}} \right) = 0\). Thus, the performance index \(J'\), as given by~\eqref{minJ'}, attains the minimum in (\(\Omega\)-\(O\))\(\cup\)\(\partial O\) when \(J\) reaches the minimum in the mapped space \(\Xi\). 
\end{proof}
According to the gradient of performance index \(J\) analyzed in Theorem \ref{Theorem27}, \(\partial J/\partial {p_i}\) is located in \({T^*_{{p_i}}}\Xi \), and \({\eta ^{ij}}\) is employed to ensure that \(\left( {\partial J/\partial {p_i}} \right){\eta ^{il}}\) is located in \({T_{{p_i}}}\Xi \), where the \(i\)-th row , \(j\)-th column element of the inverse Riemannian metric matrix  is \({\eta ^{ij}}\). Thus, \(\left( {\partial J/\partial p} \right) \cdot {\eta ^{il}}\left( p \right) = -\alpha \left( {p - {p^*}} \right)\), where \(\alpha\) is positive constant. Following  the control input \(u =  - {k_p}\left( {p - p^*} \right)\), the control input in the mapped space \(\Xi\) is designed as 
\begin{align*}
{u_i} =  - {k_p}\frac{{\partial J}}{{\partial {p_i}}}{\eta ^{il}}\left( {{p_i}} \right) 
\end{align*}
According to \eqref{deltaJ} and \({\eta _{ij}}{\eta ^{il}} = {\delta ^l_j}\), we obtain
\begin{equation}
{u_i} =  - 2{k_p}\int_{{E_i}\left( \psi  \right) - \left( {{p_i} \cup {C_{{p_i}}}} \right)} {{d_l}\left( {q,{p_i}} \right)\frac{{z_{q{p_i}}^l}}{{\sqrt {{\eta _{mn}}\left( {{p_i}} \right)z_{q{p_i}}^mz_{q{p_i}}^n} }}} \rho \left( q \right)dq,  
\label{CONTROLinput}
\end{equation}
where \(k_p\) is a scalar positive gain. As previously defined, we obtain the control input in (\(\Omega\)-\(O\))\(\cup\)\(\partial O\) as follows
\begin{equation} 
{u'_i} =  - 2{k'_p}\int_{{E'_i}\left( \psi'  \right) - \left( {{p'_i} \cup {C_{{p'_i}}}} \right)} {{d_l}\left( {{{\left( {p_i^*} \right)}^\prime },{p'_i}} \right)\frac{{z_{{{\left( {p_i^*} \right)}^\prime }{p'_i}}^l}}{{\sqrt {{{\eta '}_{mn}}\left( {{p'_i}} \right)z_{{{\left( {p_i^*} \right)}^\prime }{p'_i}}^mz_{{{\left( {p_i^*} \right)}^\prime }{p'_i}}^n} }}} \rho \left( {{{\left( {p_i^*} \right)}^\prime }} \right)d{\left( {p_i^*}\right)^\prime},
\label{Control Input}
\end{equation}
where \(k'_p\) is a scalar positive gain, and \(\ ({p_i^*})' \) represents the optimal centroid in  (\(\Omega\)-\(O\))\(\cup\)\(\partial O\) and it satisfies \(\ d({p_i^*})'=|det(J_{\tau^{-1}})|dq \). We can remove the absolute value sign after choosing an orientation of (\(\Omega\)-\(O\))\(\cup\)\(\partial O\).

\begin{remark}
Equation~\eqref{Control Input} illustrates the relationship between the control input and length metric. Then we need to prove that the designed control input can achieve collision avoidance among agents. This is accomplished by defining the minimum path that connects the agent's position to the optimal centroid using a length metric, rather than the Euclidean metric.
\end{remark}

According to the illustration of the mapping design process in Fig.~\ref{png: MAPProcess1}, we present the distributed mapping algorithm. At first, we obtain the simplicial complex \(F^j\) of \(V^j\) by utilizing the Delaunay triangulation.  We establish a harmonic diffeomorphism \(H_j\) between \(F^j\) and a unit disk \(\mathbbm{D}\). Using a quasi-conformal mapping \(\sigma_i\) to transform the unit disk \(\mathbbm{D}\) to the rectangular region \([0, L_i] \times [0,1]\). We initialize \(L_*\) to \(L_1\). For the \(j\)-th agent, by utilizing the communication mechanism of the multi-agent system, it receives ${\mu _{j - 1}}$ and ${\mu _{j + 1}}$ from \((j-1)\)-th agent and \((j+1)\)-th agent, respectively. According to ${\mu _{j-1}} > {\mu_{j + 1}}$ and ${\mu_{j}} > {\mu_{j + 1}}$, we update the parameter \(L_*=L_{j+1}\). Meanwhile, we save the Beltrami coefficient \(\mu_{j+1}\) to set \(\mu\). Then, we select the minimum value in set \(\mu\), and use the \text{LBS}~\cite{rectangular13} to obtain the optimal rectangular conformal mapping $\vartheta^*$ which minimizes the angle distortion. Secondly, for the \(t-th\) agent, it updates the point cloud $v^t\left( {x,y} \right) = v^t\left( {\left( {{L_*}/{L_t}} \right)x,y} \right)$, the rectangular conformal inverse mapping \((\vartheta _i^*)^{-1}\) maps the region \([0, L_i] \times [0,1]\) to point cloud \(V^i_*\). According to the corresponding mechanism among agents and the iterative closest point (ICP) point cloud registration~\cite{Pointcloud}, we obtain the point cloud set \(V\) in (\(\Omega\)-\(O\))\(\cup\)\(\partial O\). Using Delaunay triangulation to point cloud set \(V\), we obtain the cell decomposition \(tri\) of (\(\Omega\)-\(O\))\(\cup\)\(\partial O\). Slicing the mesh \(tri\) along a path \(v_1v_2\) yields a simply-connected open region \(\tilde \Omega \). We construct a rectangular conformal mapping \(\vartheta_i^*\) between \(\tilde \Omega \) and the rectangular region \([0, L_*] \times [0,1]\). The rectangle is subsequently mapped to an annulus using an exponential map \(\iota  ={e^{2\pi \left( {z -L_*} \right)}}\). Finally, we identify the cut vertices \(v_1v_2\) and compose the quasi-conformal map \(\varpi \) to minimize angle distortion, which forms a conformal parameterization  \(\tau = \varpi \circ \iota \circ \vartheta ^*: (\Omega - O) \cup \partial O \to \Xi\). The details of this process are illustrated in Algorithm~\ref{Algorithm1}.

\begin{algorithm}[t!]
    \renewcommand{\algorithmicrequire}{\textbf{Initialize:}}
	\renewcommand{\algorithmicensure}{\textbf{Finalize:}}
	\caption{Distributed Mapping Construction Algorithm}
        \label{Algorithm1}
    \begin{algorithmic}[1] % 
    \REQUIRE  \(N\), \(V^j\), \(H_j\), \(\sigma_j\), \(\varpi \), $\iota $; 
     \FOR{$j=1:N$}
        \STATE \(F^j\)=Delaunay(\(V^j\));
        \STATE \(v^j={\sigma _j} \circ {H_j}\left( {{V^j},{F^j}} \right)\)
        \STATE $L_*=L_1$, \(\mu\)=\(\mu_1\);
        \STATE Receive ${\mu _{j - 1}}$ and ${\mu _{j + 1}}$ from \((j-1)-th \) agent and \((j+1)-th \) agent, respectively;
        \IF{$({\mu _{j - 1}} > {\mu _{j + 1}}) and ({\mu _{j }} > {\mu _{j + 1}})$}
        \STATE Update $L_*=L_{j+1}$;
        \STATE Save \(\mu  = \{ {\mu,\mu _{j + 1}}\} \);
        \ENDIF
        \ENDFOR
        \STATE Select \({\mu _*} = \min \mu \);
        \STATE $\vartheta^*=LBS(\mu_*)$;
        \FOR{$t=1:N$}
        \STATE Update $v^t\left( {x,y} \right) = v^t\left( {\left( {{L_*}/{L_t}} \right)x,y} \right)$;
        \STATE \(v^t\) to \(t-th\) agent;
        \STATE $V^{t}=(\vartheta^*)^{-1}v^{t}$;
        \ENDFOR
        \STATE $V = ICP\left( {{V^1}, \cdots ,{V^N}} \right)$ 
        \STATE \(tri\)= \(\text{Delaunay}(V)\);
        \STATE \(\Xi\)=\(\varpi \circ \iota\circ\vartheta^* \left( {V,tri} \right)\);
    \end{algorithmic}
\end{algorithm}

In the mapped space \(\Xi\), it is noted that the final configuration of the multi-agent system is solely dependent on the partition angle. According to the proof of the existence of optimal solutions in Theorem~\ref{Theorem27} and the oriented-compatible of charts, an equivalence between \(J\left( {\psi,{\bf{p}}} \right)\) and \(J\left( {{\psi _1},{\bf{p}}} \right)\) can be established through the function \(\zeta \left( \psi  \right)\). This implies that adjusting a single phase angle is sufficient to attain the optimal performance index \(J\left( {\psi,{\bf{p}}} \right)\). From the above discussion, it is important to note that the optimal solution \({p^*_i}\) is contingent upon both \({{\psi _{i-1}}}\) and \({\psi _{i}}\), which in turn is dependent on \({\psi _{1}}\). To discretize the partition process and quantify the proximity of partitions, the following is defined: \({K^*}: = \arg {\inf _{k \in {Z^ + }}}\{ k > 0|2\pi /k \le {\varepsilon _p}\} \) with a tunable parameter \({\varepsilon _p}\), then the interval \(\left[ {0,2\pi } \right]\) can be evenly divided into $K^*$ sub-intervals~\cite{zhai23}. The index set \({\chi _k} = \{ {l_k} \in {Z^ + }|{l_k} = \arg \mathop {\inf }\limits_{i \in {I_N}} |{\psi _i} - 2\pi \left( {k - 1} \right)/{K^*}|\} \) is introduced to identify the element closest to \(2\pi \left( {k - 1} \right)/{K^*}\).  In the subsequent Algorithm~\ref{Algorithm2}, a non-convex iterative search scheme is developed, which enables MAS to approach optimal solutions to Problem~\ref{minJ'} with high accuracy. First of all, some key parameters are initialized before executing the algorithm. Secondly, we run the Distributed Mapping Algorithm~\ref{Algorithm1} to obtain the conformal mapping \(\tau\). Furthermore, when \(t < {T_\varepsilon }\) and the \(i\)-th agent is closest to the phase angle \(2\pi\left({k-1} \right)/{K^*}\), the position of the \(i\)-th agent \(p_i\) is updated according to \eqref{PIandUI} and \eqref{CONTROLinput} in the mapped space \(\Xi\), and the real agent is updated with  \eqref{PIandUI} and \eqref{Control Input} in the original space. If the \(i\)-th agent is not closest, the partition bar \({\psi_i}\) is updated according to \eqref{partitioning dynamic} and \(p_i\) is still updated with \eqref{PIandUI} and \eqref{CONTROLinput} in the mapped space \(\Xi\). Meanwhile, In (\(\Omega\)-\(O\))\(\cup\)\(\partial O\), the partition bar \({\psi' _i}\) is updated according to \eqref{Partitioning dynamic} and \(p_i\) is still updated with \eqref{PIandUI} and \eqref{Control Input}. This process continues for \( {T_\varepsilon }\). The \(i\)-th agent then saves the results in the set \(P^k_i\) and \(P'^k_i\) while creating a new set \(D^k_i\) and \(D'^k_i\). The \(i\)-th agent communicates with \((i-1)\)-th agent to obtain \(D^k_{i-1}\) and \(D'^k_{i-1}\) and checks if \(\left| {{D^k_i}} \right| \ne \left| {{D^k_i} \cup {D^k_{i - 1}}} \right|\) holds. Here \(\left| {{D^k_i}} \right|\) denotes the cardinality of the set. If the aforementioned condition is satisfied, the corresponding sets \(D^k_i\) and \(D'^k_i\) are updated with \({{D^k_i}}=  {{D^k_i} \cup {D^k_{i - 1}}} \) and \({{D'^k_i}}  ={{D'^k_i}\cup{D'^k_{i - 1}}}\) before being transmitted to the \((i+1)-th \) agent. Subsequently, the \(i\)-th agent continues to receive \({D^k_{i - 1}} \) and \({D'^k_{i - 1}} \) from the \((i-1)\)-th agent. Once the sets no longer update, the \(i\)-th agent computes \({J^k_{{i}}}=\sum\nolimits_{{D^k_i}} {{J^k_{{E_i}}}} \) and \({J'^k_{{i}}}=\sum\nolimits_{{D'^k_i}} {{J'^k_{{E'_i}}}} \). Finally, the agent selects \({k^*} = \arg {\min _{1 \le k \le {K^*}}}{J'^k_{{i}}}\) through collaboration with other agents, finalizing \({\psi' _i}\) and \(p'_i\) using \({D'^{k^*}_i}\).  To ensure the solution to Problem~\ref{minJ'} is approximated with sufficiently high accuracy, we need to analyze the proximity of the actual solution to the true optimal solution. Hence, we propose the Theorem~\ref{Theorem29} as follows.
\begin{algorithm}[t!]
    \renewcommand{\algorithmicrequire}{\textbf{Initialize:}}
	\renewcommand{\algorithmicensure}{\textbf{Finalize:}}
	\caption{Iterative Search Algorithm}
    \label{Algorithm2}
    \begin{algorithmic}[1] % 
        \REQUIRE  $K^*$ $N$, \({T_\varepsilon }\), \(p_i\), \(\psi _i\), \({k_\psi }\), \({k_p}\),   \(V^j\); 
        \STATE Run Algorithm~\ref{Algorithm1}; % Distributed Mapping Construction
        \FOR {$k=1:K^*$}
            \STATE Set $t=0$;
            \WHILE {t < \({T_\varepsilon }\)}
                \IF {$i = {\min _{{l_k} \in {\Lambda _k}}}{l_k}$}
                    \STATE Update \({p_i}\) with \eqref{PIandUI} and \eqref{CONTROLinput}
                    \STATE Update \({p'_i}\) with \({p'_i} = {\tau ^{ - 1}}\left( {{p_i}} \right)\), \eqref{PIandUI} and \eqref{Control Input}
                \ELSE
                    \STATE Update \({\psi_i}\) with \eqref{partitioning dynamic}
                    \STATE Update \({\psi' _i}\) with \({\psi '_i} = {\tau ^{ - 1}}\left( \psi  \right)\) and \eqref{Partitioning dynamic}
                    \STATE Update \({p_i}\) with \eqref{PIandUI} and \eqref{CONTROLinput}
                    \STATE Update \({p'_i}\) with \({p'_i} = {\tau ^{ - 1}}\left( {{p_i}} \right)\), \eqref{PIandUI} and \eqref{Control Input}
                \ENDIF
            \ENDWHILE
            \STATE Save \({P^k_i} = \left( {{\psi^k _i},{p^k_i},{J^k_{{E_i}}}} \right)\), \({P'^k_i} = \left( {{\psi'^k _i},{p'^k_i},{J'^k_{{E_i}}}} \right)\), \({D^k_i} = \{{J^k_{{E_i}}}\} \) and \({D'^k_i} = \{{J'^k_{{E_i}}}\} \);
            \STATE Receive \({D^k_i}\) and \({D'^k_i}\) from \((i-1)\)-th agent
            \WHILE {\(\left| {{D^k_i}} \right| \ne \left| {{D^k_i} \cup {D^k_{i - 1}}} \right|\)}
                \STATE Update \( {{D^k_i}}  =  {{D^k_i} \cup {D^k_{i - 1}}} \);
                \({{D'^k_i}}  = {{D'^k_i} \cup {D'^k_{i - 1}}} \);
                \STATE \({D^k_i}\) and \({D'^k_i}\) to \((i+1)-th\)agent
                \STATE  Receive \({D^k_{i-1}}\) and \({D'^k_{i-1}}\) from \((i-1)\)-th agent
            \ENDWHILE
            \STATE Compute \({J^k_{{i}}}=\sum\nolimits_{{D^k_i}} {{J^k_{{E_i}}}} \) and \({J'^k_{{i}}}=\sum\nolimits_{{D'^k_i}} {{J'^k_{{E'_i}}}} \)
        \ENDFOR
        \STATE Select \({k^*} = \arg {\min _{1 \le k \le {K^*}}}{J'^k_{{i}}}\)
        \STATE Finalize \({\psi' _i}\) and \(p'_i\) with \({D'^{k^*}_i}\);    
    \end{algorithmic}
\end{algorithm}

\begin{theorem}
Algorithm~\ref{Algorithm2} enables Multi-agent systems to approximate optimal solutions to Problem ~\eqref{minJ'} with arbitrary small tolerance.
\label{Theorem29}
\end{theorem}
\begin{proof}
For the sake of simplicity,  we denote \({J'\left( {\psi ',{\bf{p'}}} \right)}\) as \(I\left( {\phi ,\bf{x}} \right)\). And simplify the form of the performance index as follows
\[I\left( {\phi,{\bf{x}}} \right) = \sum\limits_{i = 1}^N {\int_{{E_i}\left( \phi  \right)} {{f}\left( {{d_l}\left( {{x_i},q} \right)} \right)\rho \left( q \right)dq} }\]
Define \({\omega _f}\left( {\theta ,{x_i}} \right) = \int_{{\varphi _{in}}}^{{\varphi _{out}}} {f\left( {{d_l}\left( {{x_i},q} \right)} \right)\rho \left( {\theta ,\varphi } \right)} {\left( {R + r\cos \varphi } \right)^2}\sin \varphi d\varphi\). Let \(I\left( {{\phi^*},{\bf{x}^ * }}\right)\) and \(I\left( {{\phi ^\varepsilon },{\bf{x}^\varepsilon }}\right)\) denote the optimal solution to Problem ~\eqref{minJ'} and actual solution that we solve, respectively. Then their error can be estimated by 
\begin{align*}
\left| {I\left( {{\phi ^*},{\bf{x}}^*} \right) - I\left( {{\phi ^\varepsilon },{\bf{x}}^\varepsilon } \right)} \right| &=  {I\left( {{\phi ^*},{\bf{x}}^k} \right) - I\left( {{\phi ^*},{\bf{x}}^k} \right) + I\left( {{\phi ^*},{\bf{x}}^k} \right) - I\left( {{\phi ^k},{\bf{x}}^k} \right) + I\left( {{\phi ^k},{\bf{x}}^k} \right) - I\left( {{\phi ^\varepsilon },{\bf{x}}^\varepsilon } \right)}\\ &\le\left|{I\left( {{\phi ^*},{\bf{x}}^*} \right) - I\left( {{\phi ^*},{\bf{x}}^k} \right)} \right| + \left| {I\left( {{\phi ^*},{\bf{x}}^k} \right) - I\left( {{\phi ^k},{\bf{x}}^k} \right)} \right| \\&
+\left| {I\left( {{\phi ^k},{\bf{x}}^k} \right) - I\left( {{\phi ^\varepsilon },\bf{x^\varepsilon }} \right)} \right|
\end{align*}
Analyze the first term in the right hand, it focuses on the variation between \({x^k}\) and \({x^*}\):
\begin{align*}
\left| {I\left( {{\phi ^*},{\bf{x}}^*} \right) - I\left( {{\phi ^*},{\bf{x}}^k} \right)} \right| &= \left| {\sum\limits_{i = 1}^N {\int_{{E_i}\left( {{\phi ^*}} \right)} {{f_i}\left( {{d_l}\left( {{x^*},q} \right)} \right) - {f_i}\left( {{d_l}\left( {{x^k},q} \right)} \right)} } \rho \left( q \right)dq} \right| \\&\le \sum\limits_{i = 1}^N {\int_{{E_i}\left( {{\phi ^*}} \right)} {\left| {{f_i}\left( {{d_l}\left( {{x^*},q} \right)} \right) - {f_i}\left( {{d_l}\left( {{x^k},q} \right)} \right)} \right|} } \rho \left( q \right)dq \\&= \sum\limits_{i = 1}^N {\int_{{E_i}\left( {{\phi ^*}} \right)} {\left| {d_{_l}^2\left( {{x^*},q} \right) - d_l^2\left( {{x^k},q} \right)} \right|} } \rho \left( q \right)dq
\end{align*}
Considering two geodesics \({\gamma _1}\left( t \right) = \left( {{\theta _1}\left( t \right),{\varphi _1}\left( t \right)} \right)\) and \({\gamma _2}\left( t \right) = \left( {{\theta _1}\left( t \right),{\varphi _1}\left( t \right)} \right)\) which denote the geodesic from \({x^k}\) to \(q\) and from \({x^*}\) to \(q\) respectively. We denote \(L\left( t \right) = \left( {{\theta _1}\left( t \right) - {\theta _2}\left( t \right),{\varphi _1}\left( t \right) - {\varphi _2}\left( t \right)} \right) = \left( {{L^\theta }\left( t \right),{L^\varphi }\left( t \right)} \right)\). Due to the region \({{E_i}\left( {{\phi ^*}} \right)}\) is a flat Riemannian manifold, the curvature tensor of \({{E_i}\left( {{\phi ^*}} \right)}\) is equal to zero and the covariant derivative of deviation vector is equal to the general derivative of deviation vector, which allows obtaining 
\[\left| {d_l^2\left( {{x^*},q} \right) - d_l^2\left( {{x^k},q} \right)} \right| = \left|\eta \left( {L,\frac{{DL}}{{dt}}} \right){\left( {\Delta t} \right)^2}\right|\]
where \(\eta \) represents the Riemannian metric of  \({{E_i}\left( {{\phi ^*}} \right)}\), \({DL}/{dt}\) is the covariant derivative of deviation vector, and \({\Delta t}\) represents the time from \(q\) to \({x^k}\) and \({p^*}\), \(\eta \left( {L,\frac{{DL}}{{dt}}} \right) = {\eta _{ij}}{L^i}{\left( {{{DL}}/{{dt}}} \right)^j}\).  In a small neighborhood of flat Riemannian manifold \({{E_i}\left( {{\phi ^*}} \right)}\), \(\sqrt {\eta \left( {L, L} \right)}  \approx {\left\| L \right\|_2}\), \(| {\theta _1}\left( t \right) - {\theta _2}\left( t \right)|\) is the biggest component of \(L\left( t \right)\). According to the equivalence of norm in finite-dimensional Euclidean Spaces, we have \(\forall {\varepsilon _t} > 0\), \(\exists {\delta _t} = {\varepsilon _t}/\sqrt 2 \), when \({\left\| {{\theta _1}\left( t \right) - {\theta _2}\left( t \right)} \right\|_2} < {\delta _t}\), \({\left\| L \right\|_2} \le \sqrt 2 {\left\| L \right\|_\infty } = \sqrt 2 {\left\| {{\theta _1}\left( t \right) - {\theta _2}\left( t \right)} \right\|_\infty } \le \sqrt 2 {\left\| {{\theta _1}\left( t \right) - {\theta _2}\left( t \right)} \right\|_2} < \sqrt 2 {\delta _t} = {\varepsilon _t}\). we denote \({\varepsilon' _t} = \sup {\varepsilon _t}\), suppose \({\varepsilon' _t} \) isn't arbitrarily small, then \(\exists {\varepsilon _0} > 0\), such that \({\varepsilon' _t} > {\varepsilon _0}\), due to \({\varepsilon' _t} \) is a supremum, \(\forall \varepsilon  > 0\), \(\exists {\varepsilon _t}\), such that \({\varepsilon' _t} - \varepsilon  < {\varepsilon _t} < {\varepsilon' _t}\), we set \(\varepsilon  = {\varepsilon _0}/2\), then \({\varepsilon' _t} - \left( {{\varepsilon _0}/2} \right) < {\varepsilon _t} < {\varepsilon' _t}\), this contradicts \(\varepsilon _t\) being an arbitrarily small distance (which should be less than \({\varepsilon _0}\)), So \({\varepsilon' _t}\) is also an arbitrarily small distance. Owing to the inequality \(|\eta \left( {L,{{DL}}/{{dt}}} \right)| \le \sqrt {\eta \left( {L,L} \right)} \sqrt {\eta \left( {{{DL}}/{{dt}},{{DL}}/{{dt}}} \right)}\le {{M_1} \cdot {M_2}} \) with \({M_1} = \sup \sqrt {\eta \left( {L,L} \right)} \) and \({M_2} = \sup \sqrt {\eta \left( {{{DL}}/{{dt}},{{DL}}/{{dt}}} \right)} \), we obtain 
\begin{align*}
\left| {I\left( {{\phi ^*},{\bf{x}}^*} \right) - I\left( {{\phi ^*},{\bf{x}}^k} \right)} \right| &\le \sum\limits_{i = 1}^N {\int_{{E_i}\left( {{\phi ^*}} \right)} {\left| {d_{_l}^2\left( {{x^*},q} \right) - d_l^2\left( {{x^k},q} \right)} \right|} } \rho \left( q \right)dq \\ &\le \sum\limits_{i = 1}^N {\int_{{E_i}\left( {{\phi ^*}} \right)} {{M_1}{M_2}{{\left( {\Delta t} \right)}^2}} } \rho \left( q \right)dq \\ &= \bar m{M_1}{M_2}{\left( {\Delta t} \right)^2\le {\varepsilon' _t}\bar m{M_2}{\left( {\Delta t} \right)^2}}.
\end{align*}
 Analyze the second term in the right hand, and it focuses on the variation between \({\phi ^k}\) and \({\phi ^ * }\). According to the work in~\cite{zhai23}, \(\int_{{\phi _0}}^{\zeta \left( {{\phi _0}} \right)} {\omega \left( \theta  \right)} d\theta  = \bar m = \int_\phi ^{\zeta \left( \phi  \right)} {\omega \left( \theta  \right)} d\theta \), then 
\(\underline{\omega} \left| {\zeta \left( \phi  \right) - \zeta \left( {{\phi _0}} \right)} \right| \le \left| {\int_{\zeta \left( {{\phi _0}} \right)}^{\zeta \left( \phi  \right)} {\omega \left( \theta  \right)} d\theta } \right| = \left| {\int_{{\phi _0}}^\phi  {\omega \left( \theta  \right)} d\theta } \right| \le \overline {\omega} \left| {\phi  - {\phi _0}} \right|\), we obtain
\(\left| {\zeta \left( \phi  \right) - \zeta \left( {{\phi _0}} \right)} \right| \le \frac{\overline{\omega} }{\underline{\omega} }\left| {\phi  - {\phi _0}} \right|\). 
Due to \(\left| {{\phi^* _i} - {\phi^k _i}} \right| = |\zeta \left( {{\phi^* _{i - 1}}} \right) - \zeta \left( {{\phi^k _{i - 1}}} \right)|\) when equitable workload partition is finished, the following statement is true
\[\left| {I\left( {{\phi ^*},{\bf{x}}^k} \right) - I\left( {{\phi ^k},{\bf{x}}^k} \right)} \right| \le {\varepsilon _p}{\overline{\omega _f}}\left( {1 + \frac{{\overline{\omega} }}{\underline{\omega} }} \right){\sum\limits_{i = 1}^N {\left( {\frac{{\overline{\omega}  }}{\underline{\omega} }} \right)} ^{i - 1}}\]
where \({ \overline{\omega _f}} = \sup {\omega _f}\), \( \overline {\omega}   = \sup \omega \), \(\underline {\omega}   = \inf \omega \), \({\varepsilon _p} = \sup{\left| {{\phi^* _2} - {\phi^k _2}} \right|}\). According to the work in~\cite{zhai23}, there exist positive constants \(c_1\) and \(c_2\) such that \(\left| {{m_i} - {m_{i - 1}}} \right| \le {c_1}{e^{ - {c_2}t}}\), \(\forall{t>0}\), \(\mathop {\lim }\limits_{t \to \infty } {\phi _i}\left( t \right) - {\phi _i}\left( 0 \right) = \int_0^{ + \infty } {{{\dot \phi }_i}\left( t \right)dt = \frac{{{k_\phi }{c_1}}}{{{c_2}}}}  <  + \infty \), which means that \(\mathop {\lim }\limits_{t \to \infty } {\phi^k _2}\left( t \right) = {\phi^* _2}\left( t \right)\), so \({\varepsilon _p}\) is an arbitrarily small positive number.  
Analyze the last term in the right hand, \(\forall \varepsilon  > 0\), \(\exists \) \({T_\varepsilon } > 0\), when \(t > {T_\varepsilon }\), \({\left| {I\left( {{\phi ^k},{{\bf{x}}^k}} \right) - I\left( {{\phi ^\varepsilon },{{\bf{x}}^\varepsilon }} \right)} \right|}<\varepsilon\), which implies
\[\mathop {\lim }\limits_{t \to \infty } I\left( {{\phi ^k},{\bf{x}}^k} \right) = I\left( {{\phi ^\varepsilon },\bf{x}^\varepsilon } \right)\]
Above all, the error estimation \(|I\left( {{\phi ^ * },{x^ * }} \right)- I\left( {{\phi ^\varepsilon },{x^\varepsilon }} \right)|\) is bounded by
\[|I\left( {{\phi ^*},{\bf{x}}^*} \right) - I\left( {{\phi ^\varepsilon },{\bf{x}}^\varepsilon } \right)| \le {\varepsilon '_t}\bar m{M_2}{\left( {\Delta t} \right)^2} + {\varepsilon _p}\overline {{\omega _f}} \left( {1 + \frac{{\overline{\omega}  }}{\underline{\omega} }} \right)\sum\limits_{i = 1}^N {{{\left( {\frac{{\overline{\omega}  }}{\underline{\omega} }} \right)}^{i-1}}} +\varepsilon\]
This completes the proof.
\end{proof}

Theorem~\ref{Theorem29} illustrates the convergence of the optimal solutions. In order to supervise the message of each sub-region, we need to prove that the control input and the dynamics of the multi-agent system would drive each agent to the optimal centroid. Hence, we propose the Theorem~\ref{Theorem38} as follows.

\begin{theorem}
Dynamical system (\ref{Control dynamic}) with control input \eqref{Control Input} ensures that each agent can reach the centroid of the subregion. 
\label{Theorem38}
\end{theorem}

\begin{proof}
According to Theorem~\ref{Theorem27}, the gradient of the coverage performance index is given by
\[{\nabla _{{p_i}}}J = 2\sum\limits_{i = 1}^N {\int_{{E_i}\left( \psi  \right)} {{d_l}\left( {{p_i},q} \right)\frac{{\partial {d_l}\left( {{p_i},q} \right)}}{{\partial {p_i}}}\rho \left( q \right)dq} },\]
where \(p_i\) is the \(i\)-th agent position, and \(q\) is the optimal centroid of region \({E_i}(\psi)\). \(q^*\) denotes the actual optimal centroid of \({E_i}(\psi)\), and \({p^*_i}\) is the optimal centroid of \({E_i}(\psi)\). In the mapped space \(\Xi\), we define \({H_{J,{p_i}}}\left( {{p_i}} \right) = {{\partial {\nabla _{{p_i}}}J\left( {{p_i}} \right)}}/{{\partial {p_i}}}\).
According to~\cite{zhai23}, if \({H_{J,{p_i}}}\left( {{p_i}} \right)\) is full rank, then \[\mathop {\lim }\limits_{t \to \infty } \left\| {{\nabla _{{p_i}}}J} \right\| = 0,\forall i \in (1,2, \cdots ,N)\]
In original space (\(\Omega\)-\(O\))\(\cup\)\(\partial O\), considering the diffeomorphism \(\tau\), we obtain \(J' \circ \tau^{-1}  = J\). Then, we set \({\tau ^{ - 1}}\left( {{p_i}} \right) = {p^*_i}\), and select a small neighborhood of \({p^ *_i }\), where \(\left( {{x^1},{x^2}} \right)\), \(\left( {{y^1},{y^2}} \right)\) are the local coordinate of \({p^ *_i }\) and \(p_i\), respectively. According to the chain rule for derivatives of composite functions, we have
 \begin{align*}
{\left( {{H_{J'}}} \right)_{ij}} &= \frac{{{\partial ^2}J'}}{{\partial {x^i}\partial {x^j}}}, 
 \frac{{\partial J}}{{\partial {y^i}}} = \sum\nolimits_{j = 1}^2 {\frac{{\partial J'}}{{\partial {x^j}}} \circ {\tau ^{ - 1}} \cdot \frac{{\partial {{\left( {{\tau ^{ - 1}}} \right)}^j}}}{{\partial {y^i}}}} \\
({H_J})_{ij} &= \frac{{{\partial ^2}J}}{{\partial {y^i}\partial {y^j}}} = \sum\nolimits_{j,l = 1}^2 {\frac{{{\partial^2}J'}}{{\partial {x^j}\partial {x^l}}} \circ {\tau ^{ - 1}} \cdot \frac{{\partial {{\left( {{\tau ^{ - 1}}} \right)}^j}}}{{\partial {y^i}}} \cdot \frac{{\partial {{\left( {{\tau ^{ - 1}}} \right)}^l}}}{{\partial {y^k}}} + \sum\nolimits_{j = 1}^2 {\frac{{\partial J'}}{{\partial {x^j}}} \circ {\tau ^{ - 1}} \cdot \frac{{{\partial ^2}{{\left( {{\tau ^{ - 1}}} \right)}^j}}}{{\partial {y^i}\partial {y^k}}}} }  
\end{align*}
then \({H_J} = {A^T}{H_{J'}}A + B\), where \({A_{ij}} = {{\partial {{\left( {{\tau ^{ - 1}}} \right)}^j}}}/{{\partial {y^i}}}\), \(B\) contains \(({{{\partial J'}}/{{\partial {x^j}}}) \circ {\tau ^{ - 1}} \cdot {{{\partial ^2}{{\left( {{\tau ^{ - 1}}} \right)}^j}}}/({{\partial {y^i}\partial {y^k}}}})\). Because \(\tau\) is a diffeomorphism, \(d{\tau ^{ - 1}}\) is a linear isomorphism, and thus \(A\) is an invertible matrix. Because the product of an invertible matrix and another matrix does not change the rank of the matrix, \(\text{rank}\left( {{A^T}{H_{J'}}A} \right) = \text{rank}\left( {{H_{J'}}} \right)\). In addition, because of the Rank-Nullity Theorem~\cite{Riemannian geometry}, \({A^T}{H_{J'}}A\) determines the major rank structure,   \(B\) is regarded as a small disturbance, \(\text{rank}(B)\) does not change the major rank structure. According to \(\text{rank}\left( {{A^T}{H_{J'}}A + B} \right) \le \text{rank}\left( {{A^T}{H_{J'}}A} \right) + \text{rank}\left( B \right)\), we obtain \(\text{rank}\left( {{H_{J'}}} \right) = \text{rank}\left( {{H_J}} \right) = 2\), so \(\mathop {\lim }\limits_{t \to \infty } \left\| {{\nabla _{{p_i}}}J'} \right\| = 0\) is true in (\(\Omega\)-\(O\))\(\cup\)\(\partial O\). In light of the gradient of performance index~\eqref{PerformanceIndex}, it follows that 
\({{\partial {d_l}\left( {{p_i},q} \right)}}/{{\partial {p_i}}}\) is non-zero. Due to 
\(\rho \left( q \right)>0\) and \(\mathop {\lim }\limits_{t \to \infty } \left\| {{\nabla _{{p_i}}}J'} \right\| = 0\), we obtain 
\[\mathop {\lim }\limits_{t \to \infty } {d_l}\left( {{p_i},{p^*_i}} \right)=0.\]
From Theorem~\ref{Theorem29}, one has \(\mathop {\lim }\limits_{t \to \infty } {\psi^k _2}\left( t \right) = {\psi^* _2}\left( t \right)\), it's easy to generalize for any \({\psi _i}\), the following statement is true \[\mathop {\lim }\limits_{t \to \infty } {p^ *_i}\left( t \right) = \mathop {\lim }\limits_{{\psi _i} \to {\psi^* _i}} {p^ *_i}\left( {{\psi _i}} \right) = q^ *\]
Above all, it satisfies that
\[0 \le \mathop {\lim }\limits_{t \to \infty } {d_l}\left( {{p_i},q^*} \right) \le \mathop {\lim }\limits_{t \to \infty } {d_l}\left( {{p_i},{p^*_i}} \right) + \mathop {\lim }\limits_{t \to \infty } {d_l}\left( {{p^*_i},q} \right) = 0,\]
which implies that \(\mathop {\lim }\limits_{t \to \infty } {d_l}\left( {{p_i},q^*} \right) = 0\). It means that each agent can move to the centroid of the sub-region.
\end{proof}

\section{Case Study}\label{CaseStudy}
Case studies are carried out to validate the proposed coverage algorithm through numerical simulations. We compare the proposed algorithm with the Voronoi partition, revealing the limitation of the latter. Furthermore, we integrate the Voronoi partition with conformal mapping, demonstrating its effectiveness in region partitioning, albeit at a slower convergence rate compared to the sectorial partition. This comparsion underscores the applicability of conformal mapping and the efficiency of sectorial partition in this complex environment. Finally, we introduce the Control Barrier Function(CBF) technology, which employs a Euclidean metric instead of a length metric, and compare the optimal paths of the two methods. The simulation results demonstrate that the proposed coverage algorithm yields shorter motion trajectories for the multi-agents, concluding that the length metric is more effective than CBF with the Euclidean metric. The findings presented in this section confirm that the sectorial partition with conformal mapping effectively accomplishes the coverage mission and achieves workload balance in (\(\Omega\)-\(O\))\(\cup\)\(\partial O\), with a rapid convergence rate.
\begin{figure}[t!]
\centering
\includegraphics[width=0.6\textwidth]{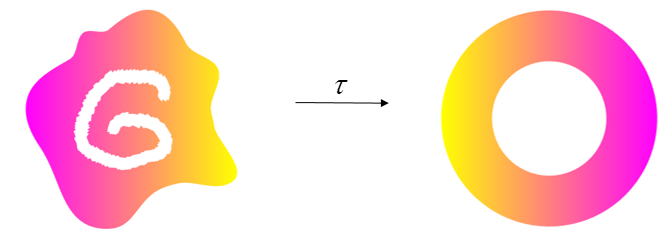}
\caption{\label{png: ColorM} The density distribution in the original space (\(\Omega\)-\(O\))\(\cup\)\(\partial O\) and the mapped region \(\Xi\). \(\tau\) is the conformal mapping. Due to the harmonic mapping would rotate the region, the density correspondence relationship would present the phenomenon.}
\label{fig3}
\end{figure}
\begin{figure}[t!]
\centering
\includegraphics[width=0.78\textwidth]{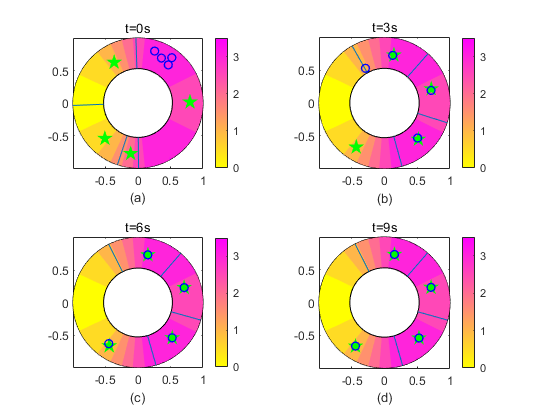}
\caption{\label{png: MappingRegionRaw} Simulation result on the mapped region \(\Xi\). Blue pointers represent the partition bar between adjacent sub-regions. Blue circles denote the mobile agents, and green stars refer to the sub-region centroids. The result illustrates the sectorial partition can finish the coverage mission and achieve workload balance in the hollow environment with a rapid convergence rate.}
\label{fig4}
\end{figure}

\begin{figure}[t!]
\centering
\includegraphics[width=0.78\textwidth]{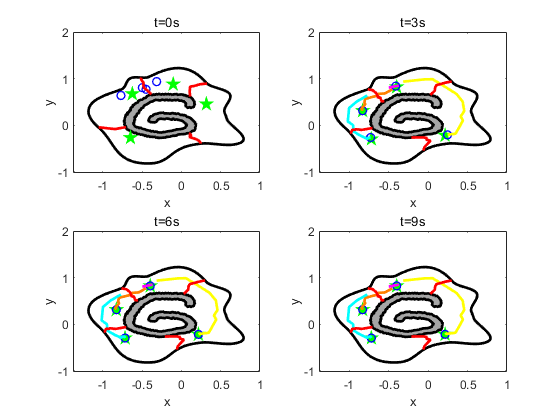}
\caption{\label{png: trajectory2} Simulation result on the original region (\(\Omega\)-\(O\))\(\cup\)\(\partial O\). Due to the bijective of the conformal mapping \(\tau\), we can find the pre-image of the partition bar, the agent's position, and the optimal centroids. The red pointers represent the pre-image of the blue pointers in \(\Xi\). The purple, orange, cyan, and yellow lines show the moving paths of the multi-agent. The phenomenon illustrates the effectiveness of conformal mapping in the coverage problem of non-star-shaped region.}
\label{fig5}
\end{figure}
\subsection{Sectorial Partition with Conformal Mapping}
\label{section4.1}
The simulation is operated according to Algorithm~\ref{Algorithm1}. In the mapped space \(\Xi\), we initialize the basic parameters as follows: $K^*=30$ $N=4$, \({T_\varepsilon }=10\), \({k_\psi=0.03 }\), \({k_p=0.1}\),  where \(v_{E'_i}\) is the coordinate message collected by \(i\)-th agent, and \(\tau\) is the designed diffeomorphism. The positions of the \(i\)-th agent, denoted as \(p_i\), are set before the algorithm runs, and \(\psi _i\) is the \(i\)-th partition bar. In Section \ref{Mappingdesign},  the value of \(L\) is minimized to reduce the conformal distortion, with \(L\) equal to 0.5293 in the numerical simulation. Based on \eqref{Riemann Metric}, we obtain the Riemannian metric of the mapped space, furthermore, we induce the length metric by the equation \eqref{Length metric}. The simulation results are presented in Fig.~\ref{fig4} and Fig.~\ref{fig5}. Particularly, Figure~\ref{fig4}  provides four snapshots of the simulation with a time interval of 4 seconds. The blue pointers represent partition bars among sub-regions,  blue circles denote the mobile agents, and green stars refer to the sub-region centroids. In addition, the color depth indicates the magnitude of workload density, with darker colors implying higher density and smaller areas occupied by high-density sub-regions due to load balancing. Figure~\ref{fig5} illustrates the simulation results on the original region. Due to the harmonic mapping rotating the surface, the original and mapped regions are mirror symmetric, with the distribution of multi-agents and optimal centroids also being mirror symmetric. The simulation results on the original space (\(\Omega\)-\(O\))\(\cup\)\(\partial O\) are shown in Fig.~\ref{fig5}, with the red pointers representing partition bars among sub-regions which are pre-image of the blue pointers and the gray part indicating the obstacle. To ensure workload balancing, the pointers are not straight lines in the original space. The orange, purple, cyan, and yellow lines represent the trajectory of agents. The meanings of blue circles, green stars, and color depth are the same as in the previous description. Experimental results demonstrate that the proposed Algorithm~\ref{Algorithm1} can fulfill the coverage mission to service events in the non-star-shaped region with the aid of conformal mapping.

\subsection{Voronoi Partition with Conformal Mapping}
To illustrate the validity of conformal mapping and the sectorial partition, this subsection will demonstrate the simulation phenomena using the Voronoi partition in the original space and using the Voronoi partition in the mapped space with conformal mapping. Finally, we will compare the result executed by the sectorial partition with conformal mapping in the subsection\eqref{section4.1}. The result of the Voronoi partition without conformal mapping in the original space is shown in the left picture of Fig.\ref{fig6}. The blue lines represent the boundary and the Voronoi partition line, the red circles denote the mobile agents, and the black region denotes the obstacle. Due to the complexity of the original space, the Voronoi partition will reduce the boundary message, and the partition line collide with the obstacle(the middle picture of Fig.\ref{fig6}). Introducing the conformal mapping \(\tau\) designed in subsection~\eqref{Mappingdesign}. Following the same idea in the subsection\eqref{section4.1}, we use the Voronoi partition in the mapped space \(\Xi\), the partition result is shown in Fig.~\ref{fig7}. Particularly, the blue lines show the Voronoi partition lines, the yellow star is the optimal centroid in the next step, the green star is the optimal centroid in this iteration step, and the red circles denote the mobile agents. The partition result illustrates the feasibility of the Voronoi partition in mapped space. According to the bijective of \(\tau\), the partition result in the original space is shown in the right of Fig.~\ref{fig6}. However, the Voronoi partition cann't achieve the workload balance(you can observe the color depth of the region), and the run time is much longer than the sectorial partition algorithm. Through the comparison, the practicability of sectorial partition with conformal mapping is demonstrated again.

\begin{figure}[t!]
\centering
\includegraphics[width=1\textwidth]{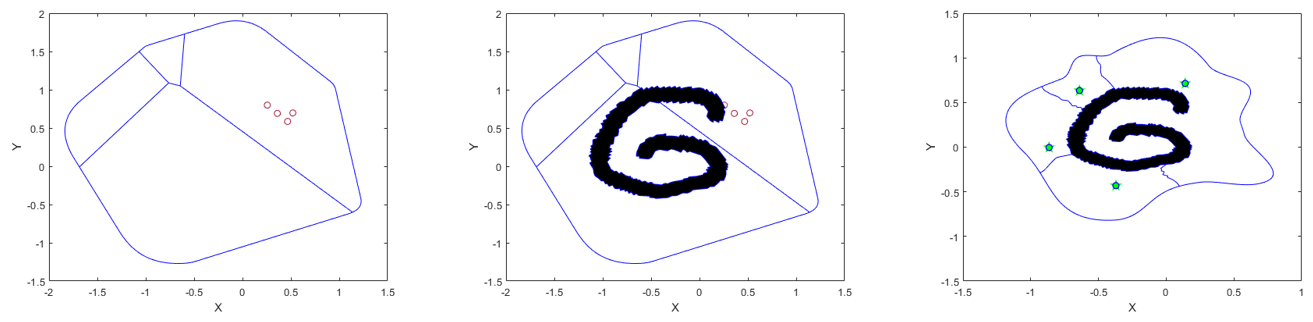}
 \caption{\label{png: voronoiii} Voronoi partition on the original region (\(\Omega\)-\(O\))\(\cup\)\(\partial O\). The left and middle snapshots illustrate the truth that the Voronoi partition cann't achieve the region partition, the partition line would collide with the obstacle. The right snapshot shows the Voronoi partition with conformal mapping in the original region (\(\Omega\)-\(O\))\(\cup\)\(\partial O\). The blue circle is the position of the agent, and the green star is the optimal centroid of the sub-region. The blue lines represent the Voronoi partition lines. It illustrates the validity of the conformal mapping, which is a powerful tool to simplify the research objects.}
\label{fig6}
\end{figure}

\begin{figure}[t!]
\centering
\includegraphics[width=0.8\textwidth]{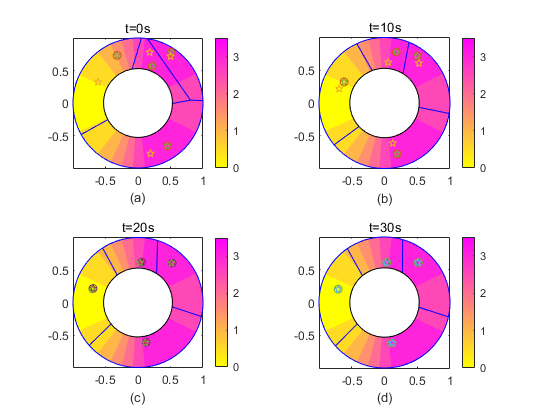}
 \caption{\label{png: MappingRegionRaw} The Voronoi partition on the mapped region \(\Xi\). Red circles denote the mobile agents, yellow stars refer to the optimal sub-region centroids in the next iteration and green stars refer to the optimal sub-region centroids in this iteration. The snapshot illustrates the validity of the Voronoi partition in the hollow environment. However, this method cann't achieve the workload balance.}
\label{fig7}
\end{figure}
\begin{figure}[t!]
\centering
\includegraphics[width=0.8\textwidth]{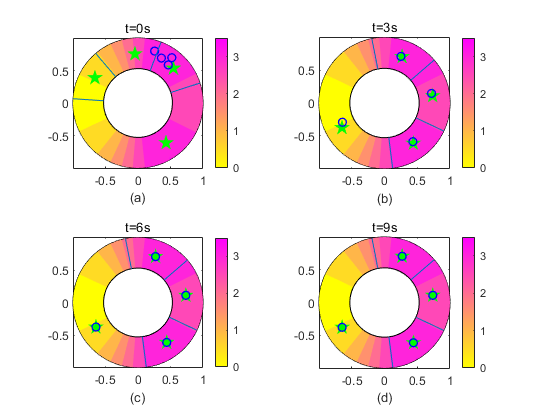}
 \caption{\label{png: MappingRegionRaw} The sectorial partition with conformal mapping and control barrier function in the mapped space \(\Xi\). Blue pointers represent the partition bar between adjacent sub-regions. Blue circles denote the mobile agents, and green stars refer to the sub-region centroids. This snapshot illustrates the validity of CBF, which replaces the length metric with the Euclidean metric.}
\label{fig8}
\end{figure}
\begin{figure}[t!]
\centering
\includegraphics[width=0.8\textwidth]{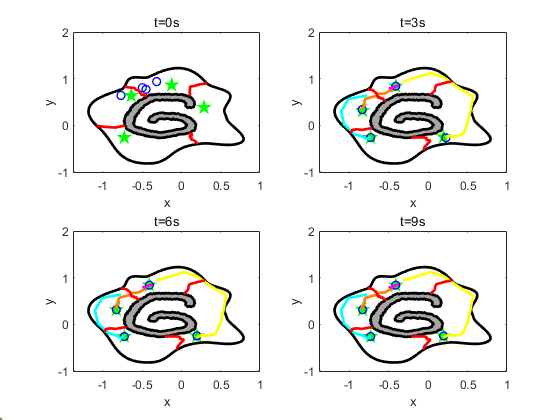}
 \caption{\label{png: MappingRegionRaw} The sectorial partition with conformal mapping and control barrier function in the original space (\(\Omega\)-\(O\))\(\cup\)\(\partial O\). The red pointers represent the pre-image of the blue pointers in \(\Xi\). The purple, orange, cyan, and yellow lines show the moving paths of the multi-agent. Those paths are longer than the paths in Fig.\ref{fig5}. This snapshot illustrates the validity and the limitations of this method.}
\label{fig9}
\end{figure}
\subsection{Control Barrier Functions with Conformal Mapping}
In this subsection, we introduce the control barrier functions, which offer the possibility of using the Euclidean metric instead of the length metric. In the mapped space \(\Xi\), inspired by the CBF application in tracking control algorithm~\cite{Avoidance24}, we adopt the CBF to modify the control law and set the obstacle as a circle, whose radius is equal to \(L\) designed in Subsection~\ref{section4.1} and the center is the origin of coordinates (the empty space enclosed by the inner boundary in Fig.~\ref{fig8}). In this process, we achieve collision avoidance by only using the Euclidean metric. The simulation illustrates that this method can fulfill the coverage requirement, and achieve workload balance. Also, the method is easier to compute. Due to the mapping correspondence relationship, the inner boundary of the pre-image is the outer boundary of the mapped space, so those paths would close to the outer boundary. However, comparing Fig.~\ref{fig9} with Fig.~\ref{fig5}, the length of the agent move path with \text{CBF} is longer than the path that uses the length metric. It means that the multi-agent system consumes more energy by using the Euclidean metric with \text{CBF} though it also can realize the coverage tasks. Explicitly, the use of the length metric is valid and effective.

\section{Conclusions}
\label{Conclusions}
This paper addressed the coverage control problem of the multi-agent system (\text{MAS}) in a non-star-shaped region. The most remarkable highlight is that a sectorial coverage formulation based on the conformal mapping which transforms a non-star-shaped region into a star-shaped one, was proposed to minimize the predefined service cost with load balancing in the mapped space. The conformal mapping being a diffeomorphism ensures that the tasks are completed with a small time deviation in the original space. A novel insight into partition regions establishes the existence of an optimum path defined by the length metric. Moreover, a non-convex iterative search algorithm was utilized to identify the optimal deployment of \text{MAS} in non-star-shaped regions with a theoretical guarantee of arbitrarily small tolerance. In addition, a distributed control strategy was applied to drive \text{MAS} toward a desired configuration automatically. Future work will focus on extending these findings to three-dimensional non-convex regions with high genus and applying the optimal transport theory to approximate the coverage control law.

\iffalse
\begin{figure}
\centering
\includegraphics[width=0.3\textwidth]{MappingRegion.png}
\caption{\label{png:MappingRegion} Mapped Workspace.}
\end{figure}
\fi

\section*{Acknowledgments}
The project is supported by the ``CUG Scholar" Scientific Research Funds at China University of Geosciences (Wuhan) (Project No. 2020138).

\end{document}